\documentclass[12pt]{article}

\usepackage[short]{optidef}
\usepackage{amsmath}
\usepackage{tikz}
\usepackage{subfigure}
\usetikzlibrary{patterns}
\usetikzlibrary{decorations.pathreplacing}
\usepackage{amssymb}
\usepackage{amsthm}
\usepackage{graphicx}
\usepackage{array}
\usepackage{comment}
\usepackage{mathrsfs}
\usepackage{enumerate}
\usepackage[top=1.25in,bottom=1.25in,left=1.25in,right=1.25in]{geometry}
\usepackage[onehalfspacing]{setspace}
\usepackage[hidelinks]{hyperref}
\usepackage{bookmark}
\usepackage{appendix}
\usepackage{bm}

\usepackage[backend=biber, style=authoryear, citestyle=authoryear-comp, uniquename=false, dashed=false, maxcitenames=3, maxbibnames=10]{biblatex}
\renewbibmacro{in:}{\ifentrytype{article}{}{\printtext{\bibstring{in}\intitlepunct}}}
\DeclareMathOperator*{\argmax}{arg\,max}

\renewcommand{\epsilon}{\varepsilon}

\theoremstyle{plain}
\newtheorem{theorem}{Theorem}
\newtheorem*{corollary}{Corollary}
\newtheorem{lemma}{Lemma}
\newtheorem{claim}{Claim}
\newtheorem{proposition}{Proposition}

\theoremstyle{definition}
\newtheorem{definition}{Definition}

\theoremstyle{remark}

\newtheorem*{example}{Example}

\newcommand{\bfr}{\mathbf{R}}
\newcommand{\bbr}{\bfr}
\newcommand{\calM}{\mathcal{M}}
\newcommand{\td}{\mathrm{d}}
\newcommand{\one}{\mathbf{1}}
\newcommand{\fb}{\mathrm{fb}}
\newcommand{\ic}{\mathrm{IC}}
\newcommand{\ob}{*}

\newcommand{\scrl}{\mathscr{L}}

\newcommand{\uw}{u}
\newcommand{\cst}{c}

\title{Collective Upkeep}
\author{Erik Madsen\footnote{Department of Economics, New York University (Email: \texttt{emadsen@nyu.edu}).} \quad Eran Shmaya\footnote{Department of Economics, Stony Brook University (Email: \texttt{eran.shmaya@stonybrook.edu}).}}
\date{January 31, 2025}

\begin{document}
\maketitle

\begin{abstract}
We design mechanisms for maintaining public goods which require periodic in-kind contributions, motivated by incentives problems facing crowd-sourced recommender systems.  Utilitarian welfare is maximized by redistributive policies which are infeasible when group members can leave or misreport their preferences.  An optimal mechanism reduces contributions for group members with low benefit-cost ratios to encourage participation; and pairs reduced contributions with restricted access to the good to ensure truthful reporting.  At most two membership tiers are offered at the optimum, indicating that ecommerce and digital content platforms may benefit substantially from offering simple user-adjustable recommendation settings.
\end{abstract}

\section{Introduction}
Social and economic groups commonly band together to provide public goods by pooling individual contributions.  A collective action problem arises whenever the group benefits of these contributions exceed their private benefits to the contributor.  In these situations, institutions are needed to coordinate contributions and ensure that group members internalize their collective benefits.

A large literature has studied the design of such institutions when group members make monetary contributions.  In some contexts, however, members may make ``in-kind'' or other nonmonetary contributions.  A leading example involves crowdsourced \emph{recommender systems} maintained by ecommerce and digital content platforms such as Netflix, Amazon, Spotify, and TikTok.  

Recommender systems aggregate user experiences with new products (which may be goods, services, or content, depending on the context) to make informed recommendations to later users.  
Experiences with untested products are often costly for the users involved, because the products may turn out to be of low quality or fit their tastes poorly. Further, some of their informational value accrues to other users, who benefit from improved recommendations.  These experiences are therefore effectively (nonmonetary) contributions to a public good, and a platform wishing to efficiently elicit them must solve a collective action problem.

Recent work in economics \parencite{crowdwisdom, recommenderCheHorner} and computer science \parencite{incentiveExploreKleinberg,bayesExploreImmorlica,bayesExploreMansour} has studied the design of recommender systems under constraints stemming from this collective action problem.  When users may disregard recommendations, platforms face dynamic constraints on the credibility of recommendations over the life cycle of an individual product.  These restrictions reflect incentive constraints in the underlying collective action problem which limit redistribution---in this case, from early to late adopters.


In this paper, we build on these insights to understand the \emph{aggregate} constraints platforms face on recommendations across many products.  Aggregate constraints are especially salient when users have little insight into the lifecycle of individual products.  For instance, Netflix subscribers may not know how long a particular movie has been on the platform or how many prior users have viewed or enjoyed it.  Instead, they may judge the credibility of the platform's recommendations primarily on their long-run quality.  
We study the design of mechanisms for providing such long-run public goods.  

We formalize our analysis through a stylized, tractable model of \emph{collective upkeep} (Section \ref{sec:model}).  In our model, a machine provides services to a large group of agents over a long time horizon.  The machine periodically breaks, reducing the value of its services until sufficient contributions have been collected to fix it.  All agents benefit from the improved services provided by a working machine but dislike contributing to fix it, with potentially heterogeneous preferences for both activities.  A designer seeking to maximize the group's long-run welfare crafts a mechanism, which regulates access to the machine's services and solicits contributions from group members. 


Under a first-best mechanism (Section~\ref{se:fb}), the designer asks all agents with sufficiently low contribution costs to contribute whenever the machine is broken.  When agents are homogeneous, individual and group preferences are aligned and the first-best mechanism is implementable.  
But when agents are heterogeneous, the first-best mechanism is redistributive and generally induces some agents to leave the group.  Participation constraints must therefore be incorporated into the mechanism design problem.

We characterize the bite of these constraints when the designer observes agents' preferences (Section~\ref{se:obedience}).  In that setting, the optimal mechanism tailors the contribution requirements of agents with binding participation constraints. 
When agent preferences satisfy a total ordering condition, optimal contributions exhibit a tripartite structure: High types contribute fully and low types contribute nothing, while intermediate types make partial contributions.  Notably, participation constraints bind only for intermediate types, and so an agent's payoff is non-monotone in his type.

In some settings, agents' preferences may be their private information and not directly observable by the designer (Section~\ref{se:ic}).  In that case, a mechanism must pair reduced contributions with restricted access to the machine in order to truthfully elicit preferences.  We show that the optimal screening mechanism offers at most two distinct participation levels.  In one, the agent enjoys full access to to the machine.  In the other, the agent enjoys only limited access while incurring a reduced responsibility for contributions.  Under this mechanism, some agents opt out of the group.

Our results can be applied to obtain new lessons for the design of recommender systems (Section~\ref{sec:fluidInterpretation}).  To this end, we show how to formally interpret our model as a large population of long-lived users who enjoy a flow of benefits from a recommender system, with breakage arising from a periodic need to learn about about new products as they are introduced and old products grow stale or are discontinued.  With this interpretation, our results point toward the desirability of user-customizable recommendation settings, which can be adjusted to obtain improved recommendations (for instance, better-personalized) at the cost of increased exposure to untested products.  Further, near-optimal performance may be achievable with only a small number of settings, as indicated by the optimality of two-choice menus in our formal model.  


Our results also provide a new perspective on existing studies of recommender system design (Section~\ref{sec:poissonInterpretation}).  To facilitate this comparison, we provide an alternative formal interpretation of our model involving a sequence of short-lived agents recommended either a safe default product or a risky alternative with renewing characteristics.  With this interpretation, our model can be viewed as a steady-state, heterogeneous-agent variant of \textcite{crowdwisdom} (hereafter KMP).  In KMP, obedience constraints limit the designer's ability to redistribute welfare from early-arriving to late-arriving agents.  Our results show that this phenomenon is pervasive: General forms of agent heterogeneity impose constraints on redistribution which bind when designing a welfare-maximizing mechanism.




Our exercise is similar in spirit to a recent literature studying mechanism design for private good provision under redistributive motives. See, for instance, \textcite{condorelli2013, akbarpour2023, paistrack} for one-sided settings; \textcite{dworczak2021redistribution} for a two-sided setting; and \textcite{kang2023} for a setting involving a mechanism run alongside an unregulated market.  In these settings, agents have differing, privately-observed social values of money, endowing a utilitarian designer with a preference for redistribution.  In our setting, agents' heterogeneous contribution costs induce an analogous motive in a public goods provision context.

\section{Model}\label{sec:model}

\subsection{Setting}
A group of agents enjoy the services provided by an excludable public good, which for concreteness we refer to as a \emph{machine}, in continuous time over a doubly infinite horizon.  At each moment the machine is either \emph{working} or \emph{broken}.  The machine provides services only when it is working, and a working machine breaks periodically. A broken machine must be fixed through contributions from group members.  Each agent can contribute only when the machine is broken and faces a limit on his rate of contribution, which without loss we normalize to 1.

Agents value the machine according to the frequency with which they use its services and contribute to fixing it.  Agents have heterogeneous preferences over these two outcome metrics, summarized by \emph{types} $\theta$ drawn from a probability measure $\lambda$ over a Borel space $\Theta$. If an agent of type $\theta$ uses the machine a fraction $r \in [0, 1]$ of the time and contributes a fraction $p \in [0, 1]$ of the time, he enjoys total utility
\begin{align} r \cdot u(\theta) - p \cdot c(\theta), \label{agentspayoff}\end{align}
where $u(\theta) > 0$ is the agent's usage benefit and $c(\theta) > 0$ is his contribution cost.  We assume that $u, c \in \mathscr{L}^1(\lambda).$  

\subsection{Design problem}\label{subsec:designProblem}

A designer oversees the machine's upkeep.  To do so, she collects contributions from agents and regulates their access to the machine's services.  Concretely, she chooses an \emph{allocation} $(R, P)$ consisting of long-run \emph{usage levels} $R: \Theta \rightarrow [0, 1]$ and \emph{contribution levels} $P : \Theta \rightarrow [0, 1]$ for all agents.  Her goal is to maximize utilitarian social welfare
\begin{equation}W(R, P) = \int \Big(R(\theta) \cdot u(\theta) - P(\theta) \cdot c(\theta)\Big) \, \lambda(\td\theta) \label{objective-full}.\end{equation}

The set of feasible allocations is limited by \emph{physical constraints} related to the process by which the machine breaks and is fixed.  We implement these constraints in a reduced-form way consistent with multiple interpretations.  Because these interpretations do not add any further formal content to our analysis, we defer describing them until Sections \ref{sec:fluidInterpretation} and \ref{sec:poissonInterpretation}.

Let $Q \in [0, 1]$ denote the long-run fraction of time the machine is working, which we refer to as the machine's \emph{uptime}.  
Then the physical constraints are given by
\begin{align}&\rho \cdot Q = \int P \, \td\lambda \label{eq:balanceConstraint}\tag{B} \\
&R(\theta) \leq Q, \quad P(\theta) \leq 1 - Q \label{eq:simplexConstraint}\tag{S-$\theta$}\end{align}


The \emph{balance condition}~\eqref{eq:balanceConstraint} says that the long-run rate at which the machine breaks (on the left-hand side) balances the long-run rate at which agents contribute to fixing it (on the right-hand side).  The parameter $\rho > 0$ can be interpreted as the machine's breakage rate.
The \emph{simplex bounds}~\eqref{eq:simplexConstraint} say that each agent's usage and contributions are bounded by the respective frequencies with which the machine is working and broken.

The set of feasible allocations may additionally be limited by \emph{economic constraints} related to agents' freedom of action.  These constraints encode familiar interim participation and incentive compatibility requirements.  Formally, the participation constraints are
\[R(\theta) \cdot u(\theta) - P(\theta) \cdot c(\theta) \geq 0. \tag{P-$\theta$} \label{eq:participationConstraint}\]
The left-hand side captures an agent's payoff from participating in the mechanism, while the lower bound of 0 is an agent's outside option obtained by never using or contributing to fix the machine. Meanwhile, the incentive-compatibility (IC) constraints are
\[R(\theta) \cdot u(\theta) - P(\theta) \cdot c(\theta) \geq R(\theta') \cdot u(\theta) - P(\theta') \cdot c(\theta). \tag{IC-$\theta\theta'$} \label{eq:ICConstraint}\]
That is, each agent can obtain any other agent's usage and contribution levels by misreporting his type.  Under incentive compatibility, no agent should prefer to do so.  

Given a feasible allocation $(R, P)$, the machine's uptime $Q$ is pinned down by the balance condition \eqref{eq:balanceConstraint}.  Nonetheless, as we will demonstrate later, it is technically convenient to include $Q$ as an additional design variable and enforce \eqref{eq:balanceConstraint} by imposing it as a constraint on the designer's welfare maximization problem.  We therefore study the designer's problem of choosing a \emph{mechanism} $(R, P, Q)$ to maximize utilitarian welfare \eqref{objective-full} subject to the physical constraints \eqref{eq:balanceConstraint} and \eqref{eq:simplexConstraint}, as well as potentially the economic constraints \eqref{eq:participationConstraint} and \eqref{eq:ICConstraint}.

\subsection{Discussion of assumptions}
We make several key modeling assumptions which merit further discussion.  

\paragraph{Utilitarian objective:} We assume that the designer is benevolent and wishes to maximize group welfare.  This objective is consistent with existing work on recommender systems in the economics literature, which similarly assumes a utilitarian objective \parencite{crowdwisdom, recommenderCheHorner}.  We additionally view this objective as appropriate for applications involving the design of recommender systems, because these systems are typically free add-ons to platforms used to sell products or content subscriptions.  To a first approximation, making recommendations maximally valuable to users facilitates a goal of maximizing product sales or subscription revenue.



Our techniques remain applicable under a more general objective in which the platform maximizes a weighted sum of group welfare and uptime or contributions.  In this more general setting, the designer's objectives conflict with group members' even when agents are homogeneous.  As a result, an additional friction impedes implementation of the designer's preferred outcome.

\paragraph{Lack of transfers:} We assume that the designer cannot subsidize agents with money to incentivize contributions.  This restriction is in line with existing work on recommender systems in both the economics and computer science literature.  It additionally accords with actual practice in most real-world platforms, where users are typically not subsidized by the platform to try recommended products.  


\paragraph{Binary state:} We assume that the machine has only two possible states.  This modeling choice, while stylized, allows us to derive microfounded constraints on feasible allocations which are tractable and plausible.  In particular, the resulting constraints respect two appealing qualitative principles: 1) the long-run average machine state is increasing in total contributions; and 2) a better long-run state relaxes constraints on usage while tightening constraints on contributions.  We therefore view a binary-state setting as useful for obtaining interpretable first-order lessons for the design of recommender systems.  

\subsection{Related models}\label{se:modelRelationships}

Our model bears important similarities to models of static public goods provision, information design, and steady-state mechanism design.  We now describe the precise relationship between these settings.  Both the similarities and differences are instructive for understanding our results and contribution to the literature, in particular the role played by our physical constraints.

\subsubsection{Public goods provision}\label{se:public}
In a static public goods provision problem, a designer must raise contributions simultaneously from a group of heterogeneous agents in order to produce a collective good.  To compare these settings and ours, we describe a benchmark static variant of our collective upkeep model. 

In the static benchmark, the group is a continuum of mass 1 with preferences described by a measure $\lambda$ over the type space $\Theta$. Let $P$ be the vector of contributions; $R$ be the vector of allocations of the public good; and $Q$ be the quantity of the public good, which is determined by total contributions through the linear production function \eqref{eq:balanceConstraint}. Agent payoffs are linear in allocations and contributions, as in \eqref{agentspayoff}. The public good is excludable, so that any $R$ such that $0\le R(\theta)\le Q$ is feasible.

Payoffs for all players are the same in the static benchmark and our model of collective upkeep.  Additionally, the constraints on feasible allocations in the former model are a subset of the ones in the latter.  In particular, the economic constraints on participation and incentive compatibility are identical in the two settings.  Furthermore, the balance condition and the simplex bound $R(\theta) \leq Q$ are required for feasibility in both settings.  The key difference between the settings is that the collective upkeep model requires the additional simplex bound $P(\theta) \leq 1 - Q.$

This extra constraint not only limits each agent's ability to contribute, but it imposes an \emph{endogenous} bound governed by overall provision of the public good.  In a static setting, exogenous contribution bounds could be justified through budget constraints or the contribution of a scarce non-monetary resource like labor. (For an example of a public goods problem in which exogenous contribution bounds are imposed, see \textcite{collectiveAction2024}.)  However, in static settings there is no natural microfoundation for contribution bounds which depend on the contributions of other agents. 
In our collective upkeep setting, by contrast, such bounds arise naturally.

The static benchmark can also be compared to the work of \textcite{hellwig2003} and \textcite{norman2004}, who design mechanisms for provisioning excludable public goods using monetary contributions.  Beside the differing nature of contributions, the main distinction between their setting and our static benchmark is that our agents are atomistic (i.e., infinitesimal).  As a result, the designer faces no aggregate uncertainty over group preferences.  Without a bound on contributions, a posted price (i.e., a fixed contribution in return for full access) is therefore optimal in our benchmark when agents' preferences are private.  This finding is consistent with the large-population limit of \textcite{norman2004}, where the optimal mechanism converges to a posted price.


\subsubsection{Information design}\label{se:information}
In a static information design problem, a sender attempts to persuade a receiver to take their preferred action by disclosing information about an underlying state. As a benchmark, consider the following problem, which is a closest approach to our collective upkeep model.  The receiver who can take one of two actions, ``Interact'' ($I$) or ``Avoid'' ($A$).  He receives payoff 0 from choosing $a = A$ and a state- and type-dependent payoff from choosing $a = I$, where the state space is $\Omega = \{W, B\}$ and he receives payoff $u(\theta)$ if $\omega = W$ and $-c(\theta)$ if $\omega = B.$  The sender's payoffs agree with the receiver's except when $(\omega, a) = (B, I)$, in which case the sender receives payoff $y - c(\theta)$ for some $y > 0.$  The state is $W$ with prior probability $Q.$ 

With observable agent types, this problem is similar to the binary-action information design problems studied in \textcite{bayesianPersuasion}.  With unobserved types, it is akin to models studied in \textcite{kolotilin2017persuasion} and \textcite{guo2019interval}.  Standard arguments from these settings imply that the binding incentive constraints involve obedience to a recommendation of $a = I$ and (if types are unobserved) truthful reporting. As a result, the incentive constraints can be reduced to the economic constraints appearing in the collective action model.  In particular, the participation constraint ensures obedience to recommendations of $a = I$, while the IC constraint ensures truthful type reporting.

The key difference between the two settings is that in the collective upkeep problem, the variables $Q$ and $y$ are endogenously determined by the balance condition.  In particular, the rate at which the receiver chooses $a = I$ when $\omega = B$ pins down $Q$ via the balance condition, and this linkage in turn determines the social value $y$ to the planner of interacting with a broken machine.  This feedback loop has no analog in static information design settings, where payoffs and state realizations are exogenous.

\subsubsection{Steady-state mechanism design}
Our setting shares important features with recent studies of long-run or steady-state mechanism design, especially in queuing \parencite{che2021optimal, margaria2023} and matching \parencite{baccara2020optimal}.  In these settings, agents view themselves as participating in a stationary system with an endogenous state whose distribution is the same for all agents.  This assumption contrasts with the prevalent classical framework in which the initial state is exogenous and agents' incentives change systematically with calendar time.  The stationary perspective plays a central role in our dynamic model interpretations. (See Sections \ref{sec:fluidInterpretation} and \ref{sec:poissonInterpretation} for details.)

\section{First-best mechanism}\label{se:fb}
We now solve for the designer's optimal mechanism satisfying the physical constraints, without imposing any economic constraints.  Formally, the designer chooses $(R, P, Q)$ to maximize aggregate welfare \eqref{objective-full} subject to the balance condition \eqref{eq:balanceConstraint} and simplex bounds \eqref{eq:simplexConstraint} for all $\theta \in \Theta.$  

\subsection{Characterizing the optimum}
An immediate observation is that welfare is increasing in agents' usage of the machine, while usage levels do not appear in the constraints.  As a result, at the optimum all agents should enjoy full access to the machine.  Formally, $R(\theta) = Q$ for all types at the optimum, saturating the corresponding simplex bounds.  Going forward, we enforce this choice of $R$ and drop it from the optimization problem.

We next establish that 
the first-best contribution levels exhibit a threshold structure: Low-cost members, whose contribution cost falls below a threshold $y^\fb$, contribute the maximum amount, while members with higher costs do not contribute at all.  Let $\bar{u} = \int \uw~\td\lambda$ be the group's aggregate usage benefit from a working machine.
\begin{proposition}\label{pr:fb}Let $y^\fb$ be the unique $y$ solving 
\begin{equation}\label{yfb}\bar u - \rho \cdot y = \int (y-\cst)^+~\td\lambda,\end{equation}
and let $W^\fb$ be the common value in~\eqref{yfb} evaluated at $y^{FB}$.  Then the mechanism $(P^\fb,Q^\fb)$ defined by
\begin{equation}\label{qfb}\rho \cdot \frac{Q^\fb}{1 - Q^\fb} = \text{Pr}_\lambda(\cst < y^\fb),\text{ and }\end{equation}
\[P^\fb(\theta)=(1-Q^\fb)\one_{[\cst < y^\fb]}(\theta)\]
solves the first-best problem and yields welfare $W^\fb$.
\end{proposition}
\noindent (All proofs omitted from the body of the paper are collected in Section \ref{se:proofs}.)

To prove the optimality of a threshold contribution structure, we consider the Lagrangian relaxation of the problem with respect to the balance condition.  Assigning a Lagrange multiplier $y$ to that constraint yields the Lagrangian
\begin{equation}\label{lagrangefb}\mathscr{L}(P,Q; y) = Q \cdot \left(\bar{u} - \rho \cdot y\right) + \int P(\theta) \cdot (y - c(\theta)) \, \lambda(\td\theta)\end{equation}
Holding $(Q, y)$ fixed, the Lagrangian factorizes over types, allowing $P$ to be optimized pointwise. Note that this factorization relies on the inclusion of $Q$ as an auxiliary optimization variable, motivating its introduction. 

The optimal multiplier $y^\fb$ measures the social value of fixing the machine. Indeed, $y^\fb$ is the amount of money the designer would be willing to pay to fix the machine immediately when it breaks down. While this idea can be made formal using dynamic programming principles, such techniques are not easily generalized to incorporate economic constraints.  By contrast, the Lagrangian approach taken in our proof is readily adapted to these constraints, as we will demonstrate shortly.

\subsection{Comparative statics}
We now examine how the optimal mechanism is shaped by the breakage rate $\rho.$  Increasing $\rho$ reduces the machine's average lifespan and decreases the social value of fixing it.  Then since $y^\fb$ tracks this social value, it follows that:
\begin{corollary}\label{co:rhoyfb}
The first-best contribution threshold $y^\fb$ is decreasing in the breakage rate $\rho$. 
\end{corollary}
\begin{proof}The left-hand side of \eqref{yfb} is decreasing in $y$ while the right-hand side is increasing. Since an increase in $\rho$ induces a downward shift in the former, $y^\fb$ must decrease in $\rho.$ \end{proof}

A lower contribution threshold means that the machine is fixed more slowly when it breaks. As a result, machine uptime must decline with $\rho$:
\begin{corollary}\label{co:rhoqfb}Machine uptime $Q^\fb$ under the first-best mechanism is nonincreasing in the breakage rate $\rho$ and decreasing whenever $Q^\fb > 0$.\end{corollary}
\begin{proof}
Equation \eqref{qfb}, which characterizes $Q^\fb$,  implies that $Q^\fb$ is nonincreasing in $\rho$ and decreasing whenever it is positive, and is nondecreasing in $y^\fb$.  Since the previous corollary established that $y^\fb$ is decreasing in $\rho,$ it follows that $Q^\fb$ is as well.
\end{proof}

In spite of these two results, the long-run aggregate contribution level $\int P \, \td\lambda$ can increase in $\rho.$  Recall that the bound on each agent's contribution level is endogenously determined by the machine's downtime $1 - Q^\fb.$  Since downtime increases in $\rho,$ all agents asked to contribute do so more often.  This force can overwhelm the countervailing reduction in $y^\fb$. 
For instance, if $\Theta = \{\theta_0\}$ is a singleton, then \eqref{yfb} implies that $y^\fb = (u(\theta_0) + c(\theta_0))/(1 + \rho) > c(\theta_0)$ for sufficiently small $\rho$.  In that regime, the aggregate contribution level is $1-Q^\fb$, which is increasing in $\rho$ in light of the previous corollary.

\section{Mechanisms with participation constraints}\label{se:obedience}
When agents are homogeneous, the first-best mechanism of a utilitarian designer delivers all agents a non-negative payoff and therefore satisfies participation.  However, when agents have heterogeneous preferences, the first-best mechanism may violate the participation constraints of some agents.  We now solve for an optimal mechanism under the additional economic constraint that all agents must prefer the mechanism to their outside option.  Formally, we impose the participation constraints \eqref{eq:participationConstraint} for all $\theta \in \Theta.$  For the time being, we continue to assume that the designer observes agents' types.

\subsection{Characterizing the optimum}

As in the first-best solution, there is no reason to restrict any agent's access to a working machine.  Indeed, such restrictions would lower social welfare while simultaneously tightening participation constraints. Hence $R(\theta) = Q$ for all types at the optimum.  Going forward, we enforce this choice of $R$ and drop it from the optimization problem.

In this reduced problem, the participation constraint may be rearranged to obtain an upper bound on each agent's contribution level:
\[P(\theta) \leq Q \cdot \nu(\theta),\]
where \[\nu(\theta) = \uw(\theta)/\cst(\theta)\] 
captures an agent's rate of substitution between usage and contributions.  We will refer to this quantity as the \emph{valuation} of an agent of type $\theta.$    This representation allows the two constraints \eqref{eq:simplexConstraint} and \eqref{eq:participationConstraint} to be combined into the single constraint \begin{align}
P(\theta) \leq \min(1 - Q, Q \cdot \nu(\theta)) \tag{SP-$\theta$} \label{eq:participationConstraintModified}
\end{align}

Since participation constraints merely tighten the upper bound on feasible contribution levels, a threshold contribution structure is optimal for the same reasons as in the first-best problem.  The new feature of the problem with participation constraints is that agents who are asked to contribute may do so at different levels.  Since the constraints \eqref{eq:participationConstraintModified} are type-dependent, types for whom $\nu(\theta) < (1 - Q)/Q$ contribute a type-dependent amount less than their feasible maximum.  These are precisely the agent types for whom participation constraints bind.

Let
\begin{equation}\label{elldef}\ell(Q; y) = Q \cdot (\bar u - \rho \cdot y)+\int \min(1-Q,Q \cdot \nu(\theta))(y-\cst(\theta))^+~\lambda(\td\theta)\end{equation}
be the designer's reduced Lagrangian under the optimal threshold contribution rule, where recall that $y$ is the Lagrange multiplier on the balance condition.  The following theorem establishes that an optimal mechanism under participation constraints can be obtained by solving the minimax problem associated with $\ell$.
\begin{theorem}\label{th:obedience}
Let $W^\ob$ be the value and $(Q^\ob, y^\ob)$ be the smallest\footnote{The selection of the smallest saddle point is not needed for optimality, but we adopt it as a convention for comparability with the mechanism characterized in Proposition \ref{pr:fb}.} saddle point of the concave-convex minimax problem
\[\min_{y \geq 0}\max_{Q \in [0, 1]} \ell(Q; y)=\max_{Q \in [0, 1]}\min_{y \geq 0} \ell(Q; y),\]
with $y^\ob$ defined as $+\infty$ if the minimum is not attained. 
Then the mechanism $(P^\ob, Q^\ob)$ with
\[P^\ob(\theta)=\min(1-Q^\ob, Q^\ob \cdot \nu(\theta))\one_{[\cst < y^\ob]}(\theta)\]
solves the designer's problem under participation constraints and delivers welfare $W^\ob.$
Moreover, $y^\ob \geq y^\fb$, and this inequality is strict if $W^\ob < W^{\fb}$.
\end{theorem}

The cost threshold $y^\ob$ is the optimized value of the Lagrange multiplier on the balance condition.  As a result, it has a natural interpretation as the (marginal) shadow value of fixing the machine.  When participation constraints bind, this value is larger than under the first-best mechanism.  Indeed, fixing the machine not only generates usage benefits for other agents (as in the unconstrained problem) but additionally loosens their participation constraints by raising the machine's uptime.  As a result, the designer can concentrate contributions on a narrower set of low-cost agents.

Surprisingly, this latter force can be powerful enough to induce more frequent repairs, and therefore higher uptime, under participation constraints than without them. (Of course, this outcome cannot improve social welfare over the first-best outcome.)  The following example illustrates this possibility.

\begin{example}
Consider a three-type environment $\Theta = \{L, M, H\}$ with
\begin{center}
\begin{tabular}{c|ccc}
& L & M & H \\
\hline
$u$ & 3 & 4 & 10 \\
$c$ & 3 & 2 & 1.25 \\
$\lambda$ & 1/3 & 1/3 & 1/3
\end{tabular}
\end{center}
The machine breaks at rate $\rho = 5.5.$  In this setting, $y^{FB} = 2.7$ and $Q^{FB} = 4/15,$ while $y^\ob = 45/14$ and $Q^\ob = 2/7 > Q^{FB}.$
\end{example}

\subsection{Ordered types}
In general, the set of participation-constrained agents are those with low contribution costs and valuations.  When costs and benefits comove, this set takes a particularly simple form.
\begin{definition}
Agent types are \emph{ordered} if $\Theta = \bfr_-$ and $\nu$ is increasing in $\theta$ while $c(\theta) = -\theta$.
\end{definition}
The substantive content of ordered types is that the type structure is one-dimensional and valuations are increasing while costs are decreasing along this dimension.  The identification of $\theta$ with the agent's cost is simply a convenient normalization.  It is without loss, since types can always be relabeled to achieve this identity.

Recall that the first-best mechanism requires agents with sufficiently low costs to contribute. For ordered types $c$ is decreasing, and therefore the first-best mechanism can be characterized by a type threshold $\theta^\fb = -y^\fb$ above which agents are asked to contribute.  
Meanwhile, because $\nu$ is increasing under ordered types, participation constraints are active on an interval of intermediate types.  Indeed, the designer's incentives are aligned with both high-type agents (who don't mind contributing) and low-type ones (for whom contributing is too costly to be socially valuable).  A conflict arises only for intermediate-type agents, who generate a large social value from contributing but do not sufficiently enjoy using the machine.  

The following proposition formalizes this reasoning.  It additionally establishes that the interval on which participation constraints bind includes the first-best contribution threshold.  

\begin{proposition}\label{prop:orderedTypesParticipation}
Suppose that types are ordered and first-best welfare is not achievable under participation constraints.  Then there exists an optimal mechanism in which participation constraints are active on an interval of types containing $\theta^\fb.$
\end{proposition}

The fact that participation constraints bind only for intermediate types is a noteworthy feature of the optimal mechanism.  It is known that these may be the binding constraints in screening and trading mechanism design problems when outside options are type-dependent \parencite{jullienOO, cramton1987dissolving}.  By contrast, in our setting the result arises even though outside options are type-independent.


\subsection{Comparative statics}
Unlike the first-best solution, the optimal threshold $y^\ob$ under participation constraints does not exhibit unambiguous comparative statics in the breakage rate $\rho$.  There are two contrasting forces here. On the one hand, as in the first-best problem, frequent breakdowns reduce the usage benefits generated by fixing the machine.  On the other hand, participation constraints tighten as breakdowns become frequent, and fixing the machine more often loosens these constraints.  Fixing the machine therefore allows the designer to concentrate contributions on lower-cost agents, yielding an aggregate cost savings.

The tension between these forces is illustrated by the limiting behavior of $y^\ob$ as $\rho$ grows large.  Depending on the thickness of the right tail of the valuation distribution, contributions may become narrowly concentrated on low-cost types (i.e., the first force dominates) or else dispersed throughout the population of agents (the second force dominates).

\begin{proposition}\label{prop:compStatYP}
Suppose that $\mathbb{E}_\lambda[\nu] = \infty$ and the distribution of $c$ under $\lambda$ has full support on $\mathbb{R}_+.$

\begin{itemize}
\item If $\lim_{v \rightarrow \infty} v \Pr_\lambda(\nu > v) = 0$, then $\lim_{\rho \rightarrow \infty} y^\ob = \infty.$

\item If $\lim_{v \rightarrow \infty} v \Pr_\lambda(\nu > v) = \infty$ and $\nu$ and $c$ are negatively associated under $\lambda$, then $\lim_{\rho \rightarrow \infty} y^\ob = 0.$
\end{itemize}
\end{proposition}

The condition that $\nu$ have infinite expectation ensures that the optimal mechanism under participation constraints generates positive welfare even for large $\rho.$ It is not a particularly restrictive assumption given that $\nu$ is a ratio of usage benefits to costs.  It is satisfied, for instance, if $\liminf_{\varepsilon \rightarrow 0} \Pr_\lambda(c < \varepsilon)/\varepsilon > 0$ and usage benefits are bounded away from zero.  The assumption is also compatible with both tail conditions appearing in the proposition.  For instance, if $\Pr_\lambda(\nu > v) \sim 1/(v \cdot (\log v)^\alpha)$ for large $v$ for some $\alpha \geq 1$, then $\mathbb{E}_\lambda[\nu] = \infty$ while $\lim_{v \rightarrow \infty} v \Pr_\lambda(\nu > v) = 0$.  On the other hand, if $\Pr_\lambda(\nu > v) \sim 1/v^\alpha$ for large $v$ for some $\alpha \in (0, 1)$, then $\mathbb{E}_\lambda[\nu] = \infty$ while $\lim_{v \rightarrow \infty} v \Pr_\lambda(\nu > v) = \infty$.

When types are ordered, the tail conditions on $\nu$ can be recast in terms of the density of costs at zero.  This is because under ordered types, valuations are high when $\theta$ is large, which is also when costs are low.  As a result, a thick right tail is equivalent to an unbounded cost density near zero, while conversely a thin right tail is equivalent to a vanishing cost density.

\begin{proposition}\label{prop:compStatYPOrdered}
Under the assumptions of Proposition \ref{prop:compStatYP}, suppose additionally that types are ordered and $\limsup_{\theta \rightarrow 0} u(\theta) < \infty$ while $\liminf_{\theta \rightarrow 0} u(\theta) > 0.$

\begin{itemize}
\item If $\lim_{\varepsilon \rightarrow 0} Pr_\lambda(c < \varepsilon)/\varepsilon = 0$, then $\lim_{\rho \rightarrow \infty} y^\ob = \infty.$

\item If $\lim_{\varepsilon \rightarrow 0} \Pr_\lambda(c < \varepsilon)/\varepsilon = \infty$, then $\lim_{\rho \rightarrow \infty} y^\ob = 0.$
\end{itemize}
\end{proposition}

The machine's uptime is a function both of the breakage rate and of aggregate contribution levels.  As discussed in our analysis of the first-best mechanism, aggregate contributions need not shrink with $\rho$.  Nonetheless, it turns out that contribution levels never grow fast enough to outweigh the higher breakage rate, and the machine's uptime decreases in $\rho.$ 

\begin{proposition}\label{qdaggerroh} $Q^\ob$ is nonincreasing in $\rho$.\end{proposition}

While the assertion is intuitive, we are not aware of any standard results regarding comparative statics of saddle points which directly imply it. The difficulty arises from the fact that the function $\ell$ of which $(Q^\ob, y^\ob)$ is a saddle point is neither sub- nor supermodular in $(Q, y, \rho)$. To prove the result, we identify an appropriate modification of $\ell$ which is submodular and apply standard results to that function.

\section{Incentive-compatible mechanisms}\label{se:ic}
We now solve for an optimal mechanism when the designer must induce agents not just to participate, but also to truthfully report their preferences.  Formally, we impose the participation constraints \eqref{eq:participationConstraint} as well as the IC constraints \eqref{eq:ICConstraint} for all $\theta, \theta' \in \Theta.$  We refer to a mechanism which satisfies both sets of constraints as a \emph{screening} mechanism.

As we have seen, when agents' types are observable, there is no benefit to restricting any agent's access to the machine.  However, when agents must self-report their contribution costs, access restrictions become crucial for screening.  The usage levels $R$ must therefore be designed in tandem with the contribution levels $P$.

The following result characterizes the main qualitative features of an optimal screening mechanism.
\begin{theorem}\label{th:ic} There exists an optimal screening mechanism $(R^*, P^*, Q^*)$ which satisfies
\[\left(R^*(\theta), P^*(\theta)\right)=\begin{cases}(Q^*, 1-Q^*), &\text{ if }\bar\nu\le \nu(\theta)\\(\hat R,\hat P),&\text{ if }\underline\nu\le \nu(\theta) < \bar\nu\\(0,0),&\text { if }\nu(\theta) < \underline\nu\end{cases}\]
for some $0 < \underline\nu \leq \bar\nu\le\infty$ and $0 < \hat R \leq Q^*$ and $0 < \hat P \leq 1 - Q^*.$  Moreover, if $\bar\nu=\infty$ then $\hat{R} = Q^*$.
\end{theorem}

In words, the designer screens agents by offering up to two distinct membership tiers: a full-access tier with high contribution requirements, and potentially also a partial-access tier with reduced contributions.  Whenever multiple membership tiers are optimal, the high tier demands the maximum contribution from all members in that tier.  Finally, all agents have the ability to opt out of the group and receive the $(0, 0)$ allocation.  Unlike the case of observable preferences, in general some agents do opt out of the group under an optimal screening mechanism.

The proof of this result relies on a reduction of the mechanism design problem to a Myersonian monopoly-sale problem.  In this reduction, buyers have quasilinear utility with valuations $\nu(\theta)$ for the good.  The quantity $R(\theta)$ becomes the probability with which a buyer of type $\theta$ receives the good while $P(\theta)$ is the price he pays for it.  The seller's objective is a weighted average of revenue, reflecting a Lagrange multiplier on the balance condition, and the utility of each buyer type.  Unlike a classic monopoly-sale problem, the simplex bound $P(\theta) \leq 1 - Q$ imposes an upper bound on the price that each buyer can pay for the good.  As a result, a posted-price mechanism may not be optimal and fractional allocations can be beneficial for screening buyers. (We provide an analysis of the monopoly-sale problem under a payment bound in Appendix \ref{singlegood}.)

We show that an optimal allocation rule corresponds to an extreme point in the set distribution functions which satisfy a moment inequality equivalent to the budget constraint.  In general, extreme points may involve one additional point of support for each moment inequality, implying our result.  
The same argument drives the findings of \textcite{le2019persuasion} and \textcite{doval2024constrained}, who show that in an information design context, each additional constraint adds up to one additional point of support in the set of posteriors induced by an optimal mechanism.  


Theorem \ref{th:ic} ensures that no more than two membership tiers are required at the optimum.  The following example exhibits a setting in which exactly two tiers are required.

\begin{example}Suppose that $\Theta = \{L, M, H\}$, with equal probability of each type. All types have contribution cost $c(\theta) = 1$, while usage utilities are $u(H)=5, u(M)=1, u(L)=0.1$. The breakage rate is $\rho=1$. 

Under the first-best mechanism, $Q^\fb =0.5$, and all types fully contribute, so that $(R^\fb(\theta), P^\fb(\theta)) = 0.5$ for all types.

Under the optimal screening mechanism, $Q^*=0.3$. Type $H$ receives full access in return for a full contribution, type $M$ receives partial access in return for a reduced contribution, and type $L$ leaves the group. The optimal allocations are $R^*(H)=0.3, R^*(M)=0.2, R^*(L)=0$ and $P^*(H)=0.7, P^*(M)=0.2, P^*(L)=0$.\end{example}

\section{Model interpretation: Fluid interaction with long-lived agents}\label{sec:fluidInterpretation}

In this section, we show how our reduced-form model can be interpreted as a club comprised of a large population of long-lived agents, who use and contribute to the machine in an incremental manner.  We describe this interpretation and derive the reduced-form physical constraints from it in Section \ref{subsec:fluidMicrofoundation}.  We then relate it to the design of crowdsourced recommendation platforms in Section \ref{subsec:fluidApplication}.

\subsection{Microfounded model}\label{subsec:fluidMicrofoundation}

The group is comprised of a mass of atomistic (i.e., infinitesimal), long-lived agents. Agents live for a random period of time drawn from a (type-dependent) distribution $F(\cdot\mid\theta)$ with expectation $1/\kappa$ for each type. (For simplicity, we assume that $\kappa$ is the same for all types.)   Deaths are balanced by an inflow of new agents at rate $\kappa.$  We assume that the population of agents is in steady state, so that the empirical distribution of agent ages $G(\cdot\mid\theta)$ is time-invariant 
with density $g(z\mid\theta)=\kappa \cdot (1-F(z\mid\theta))$.  Agent types are drawn from the distribution $\lambda$. Because all types have the same expected lifespan, $\lambda$ is also the empirical distribution of agent types in steady state. 

At each moment when the machine is working, each agent can use its services, enjoying a flow benefit $u(\theta)$ from (full) access.  At each moment when the machine is broken, each agent can contribute to fixing it, incurring a flow cost $c(\theta)$ from (fully) contributing.  Agents evaluate the payoff of participating in the mechanism according to the expected cumulative usage benefits minus contribution costs over their lifespan.

When the machine breaks, a random \emph{contribution quantum} (i.e., a cumulative contribution aggregated across agents and time) is required to fix it.  The size of this quantum is drawn from a fixed distribution with expectation normalized to 1, independent of the machine's history.  When the machine is fixed, it breaks after a random \emph{lifespan} drawn from a fixed distribution with expectation $1/\rho$, independent of the machine's history.

A mechanism assigns personalized (type- and age-dependent) usage and contribution levels, varying based on the current machine state as well as potentially the history of the system.  We restrict attention to mechanisms that induce a stationary environment for agents born at different calendar times. Formally, a mechanism must be compatible with a stationary \emph{system process} $X_t = (\omega_t, r_t, p_t),$ where:
\begin{itemize}
    \item $\omega_t \in \{W, B\}$ is the state of the machine at time $t$
    \item $r_t:\Theta\times \bfr_+ \to [0,1]$  records the time-$t$ usage level $r_t(\theta,z)$ of each agent of type $\theta$ and age $z$
    \item $p_t:\Theta\times \bfr_+ \to [0,1]$ records the time-$t$ contribution level $p_t(\theta,z)$ of each agent
\end{itemize}
Note that the information in $X$ is sufficient to derive the entire history of the system, including the lifespans of the machine (using the transition times of $\omega$ from W to B) and the contribution quanta (using the transition times of $\omega$ from B to W plus the contribution levels $r$).  

Given any stationary system process, let 
$\bar{r}_t(\theta) = \int r_t(\theta, z)g(z \mid \theta)~\td\theta$ and $\bar{p}_t(\theta) = \int p_t(\theta, z) g(z \mid \theta)~\td\theta$
be the group-average usage and contribution levels for each type.
 The associated \emph{reduced form} $(R, P, Q)$ is well-defined, where $R(\theta)= \mathbb{E}\bar{r}_t(\theta), P(\theta) =\mathbb{E}\bar{p}_t(\theta), Q=\Pr(\omega_t=W)$ and each of these expressions is independent of $t$ by stationarity.  
 
 The normalized expected  payoff for an agent of type $\theta$ born at time $t$ can be written in terms of this reduced form as 
\[\kappa \cdot \mathbb{E}\int_0^\infty \td F(\tau \mid \theta)\int_0^\tau\Big(r_{t+z}(\theta,z) \cdot u(\theta)  -p_{t+z}(\theta,z) \cdot c(\theta)\Big) \, \td z =R(\theta) \cdot u(\theta) - P(\theta) \cdot c(\theta),\]
as in~\eqref{agentspayoff}, where the inner integral aggregates the agent's flow payoffs over his lifespan and the outer one averages over possible lifespans.  It follows that utilitarian welfare, defined as the average welfare of a randomly-selected agent, equals \eqref{objective-full}. (By the ergodic theorem, this objective is equivalent to an alternative specification of welfare as the long-run average flow payoffs of group members.)

\emph{Markovian} mechanisms condition current usage and contribution levels only on the current state of the machine.  Formally, a Markovian mechanism is characterized by a Borel function $\sigma:\Theta\times\{W,B\}\to [0,1]$, with the interpretation that an agent of type $\theta$ is assigned usage level $\sigma(\theta,W)$ when the machine is working and contribution level $\sigma(\theta,B)$ when the machine is broken.  We do not formally describe more general mechanisms, as we will show that any system process satisfying certain necessary conditions is payoff-equivalent to one induced by a Markovian mechanism.

\begin{definition}~\label{as:fluid}
A stationary system process is \emph{admissible} if:
\begin{enumerate}
\item Machine lifespans are i.i.d.\ with expectation $1/\rho$.
\item Contribution quanta are i.i.d.\ with expectation $1$.
\item $r_t(\theta)\le \one_{[\omega_t=W]}$ and $p_t(\theta) \le \one_{[\omega_t=B]}$ for all $t$.
\end{enumerate}\end{definition}
Every stationary system process compatible with some (not necessarily Markovian) mechanism is admissible.  The first two requirements of admissibility ensure that the process respects the independence of successive lifespans and cumulative contributions.  The last requirement ensures that the machine is used only when it is working and contributions are collected only when the machine is broken.

The following result establishes that every reduced form of an admissible system process satisfies the physical constraints.  Furthermore, all triples $(R, P, Q)$ satisfying the physical constraints can be achieved via a Markovian mechanism.  In particular, the physical constraints are necessary and sufficient for feasibility of an allocation.

\begin{proposition}\label{pr:physfluid}\begin{enumerate}
\item Fix an admissible system process. Then its reduced form satisfies the physical constraints.

\item Suppose that the triple $(R, P, Q)$ satisfies the physical constraints.  Then there exists a Markovian mechanism $\sigma$ and a stationary system process $X$ such that $\sigma$ is compatible with $X$ and $(R, P, Q)$ is the reduced form of $X$.
\end{enumerate}
\end{proposition}

\subsection{Application: Crowdsourced recommendation platforms}\label{subsec:fluidApplication}

We now discuss how to apply the fluid interpretation of our model as a stylized model of ecommerce and digital content platforms such as Netflix, Amazon, TikTok, and Tiktok, which rely on user experiences with new products to train their recommendation algorithms.  

In this application, the agents are users of the platform, while the machine is a recommendation algorithm which makes personalized product recommendations to users.  The machine is working when it is informed about the attributes of the current product inventory.  The machine breaks when the platform's inventory turns over and old products are replaced by new ones with unknown attributes.  It is fixed by recommending new products to users and monitoring their experiences.  

Inventory turnover could occur for a variety of reasons.  For instance, user tastes may change over time; old products may be replaced by new versions; and the mix of products available on the platform may change.  As an example of the first possibility, Spotify must re-learn the attributes of its music library over time as trends change and the popularity of particular artists changes.  As an example of the second, Amazon needs to learn the attributes of each new generation of wireless headphones as technology advances.  As an example of the third, Netflix periodically rotates the library of content to which it holds streaming licenses, and it must learn the attributes of newly-licensed content.

Platforms monitor user experiences through several channels.  One avenue is active feedback from product ratings.  An additional source is passive feedback from observable user behavior.  For instance, Netflix tracks behavior such as whether users finished a recommended television show, how many episodes they watched, their viewing time per session, and whether they subsequently viewed related shows.  Both active and passive feedback may be costly for users to provide, the former due to the time and energy required to provide ratings, and the latter due to the potential for a poor experience with an untested product.

Usage of the machine corresponds to informative, personalized recommendations about well-understood products.  Access is controlled by modulating the informativeness of recommendations, for instance by mixing in recommendations for products with which the user is already familiar.  Contributions correspond to recommendations to try untested products.  

A user's commitment to contribute in return for access can be implemented by bundling recommendations for tested and untested products.  In other words, a user can simply be shown a stream of recommendations from the platform, with no indication of whether a particular recommendation is intended to explore a new product or exploit the engine's knowledge of an older product.  So long as product experiences are sufficiently noisy or observed with delay, users can accurately assess only the long-run quality of recommendations but not the likely purpose of any particular one.  In those environments, the ex interim participation constraints we impose in our model are the binding ones.


This application can be formally mapped onto our abstract model in the following way.  Uncertainty about the available products on the platform is summarized by a state $\omega \in \Omega$ pertaining to an inventory $X.$  Each user receives a flow payoff $v(\theta, x, \omega)$ from trying product $x \in X$ from the library when the library's state is $\omega.$  When $\omega$ is known, the algorithm recommends $x^*(\theta, \omega) \in \argmax_{x \in X} v(\theta, x, \omega),$ generating a usage benefit $u(\theta) = v(\theta, x^*(\theta, \omega), \omega)$ for each user.  Since consumers may return to the platform repeatedly over time, a choice $x$ from the inventory should be interpreted as a bundle of individual products.

When the machine breaks, $\omega$ is drawn anew from some known distribution $\mu^\omega$.  (For notational simplicity, the inventory $X$ is held fixed, with the understanding that elements of $X$ are merely labels that carry no informational content.)  A mass of products of random size must be consumed by users before $\omega$ is revealed, at which point the algorithm can once again make informed recommendations.  

While the algorithm is learning $\omega,$ users are recommended products from the library according to some (potentially type-dependent) distribution $\mu^{x}(\theta)$. Hence in the learning regime, users receive a usage benefit $v^0(\theta) = \mathbb{E}_{\omega, x}(v(\theta, x, \omega) \mid \theta) < 0$. 
Users may additionally incur direct costs of providing explicit feedback such as ratings, which can be captured by an additional flow cost $c^0 \geq 0$ of providing ratings.  
Each user's effective cost of contributing is $c(\theta) = c^0 -v^0(\theta) > 0.$


\section{Model interpretation: Sequential interaction with short-lived agents}\label{sec:poissonInterpretation}

In this section, we show how our model can be interpreted as a sequence of interactions with a population of short-lived agents, who each use and contribute to the machine exactly once.  We describe this interpretation and derive the reduced-form constraints from it in Section \ref{subsec:poissonMicrofoundation}.  We then relate it to the design of incentive-aware bandit exploration models in Section \ref{subsec:poissonApplication}.

\subsection{Microfounded model} \label{subsec:poissonMicrofoundation}
Short-lived agents arrive according to a Poisson process with rate normalized to $1$. Arriving agents' types are drawn i.i.d.\ from the distribution $\lambda$.  An arriving agent can use the machine if it is working and can contribute to fix the machine if it is broken, after which he departs forever.  Using a working machine gives an agent of type $\theta$ usage benefit $u(\theta),$ while contributing to fix a broken machine incurs a cost $c(\theta).$  A contribution fixes the machine with a probability we normalize to 1.  When the machine is fixed, it breaks after a random lifespan drawn from a fixed distribution with expectation $1/\rho$, independent of the machine's history.

As in our fluid interpretation (section \ref{sec:fluidInterpretation}), we focus on mechanisms which are compatible with a stationary system process.  In a sequential arrival environment, the system process tracks agent arrivals, their types, the actions they took, and the state of the machine.  Every stationary system process induces a well-defined reduced form $(R, P, Q),$ where:
\begin{itemize}
\item $Q$ is the long-run proportion of time that the machine is working. Equivalently, it is the probability that the machine is working at any given time. 
\item $P:\Theta\to [0,1]$ is the long-run density of machine usage.  That is, for every Borel set $B \subset \Theta$, the long-run frequency of agents who are of types $\theta\in B$ and use the machine is $\int_B P~\td\lambda$. 
\item $R:\Theta\to [0,1]$ is the long-run density of contributions.  That is, for every Borel set $B \subset \Theta$, the long-run frequency of agents who are of types $\theta\in B$ and contribute is $\int_B R~\td\lambda$. 
\end{itemize}

Given these definitions, the expected payoff for an agent of type $\theta$ who is born at type $t$ can be written in terms of this reduced form as in \eqref{agentspayoff}, and the long-run expected welfare of an arriving agent is as in \eqref{objective-full}.

Markovian mechanisms stochastically direct arriving agents to use or fix the machine with probabilities that depend only on their type and the current state of the machine.  Formally, a Markovian mechanism is characterized by a Borel function $\sigma:\Theta\times\{W,B\}\to [0,1]$, with the interpretation that an arriving agent of type $\theta$ uses the machine with probability $\sigma(\theta,W)$ if it is working and fixes it with probability $\sigma(\theta,B)$ if it is broken, with independent randomization across agents.

As in our analysis of the fluid interpretation, we explicitly impose only a few necessary conditions on admissible system processes.  These conditions will imply the physical constraints and ensure implementability using Markovian mechanisms.

\begin{definition}~\label{as:poisson}
A system process is \emph{admissible} if:
\begin{enumerate}\item Conditional on the history of the system process before time $0$, agent arrivals after time 0 follow a Poisson process with arrival rate $1$ and arriving agent types are i.i.d.\ with distribution $\lambda$.
\item Machine lifespans are i.i.d.\ with expectation $1/\rho$.
\item Arriving agents use the machine only when it is working and contribute only when it is broken.
\end{enumerate}\end{definition}
The first requirement ensures ensures that the agent arrival process is exogenous, and in particular that it is independent of the machine state. This assumption can also be viewed as a nonanticipation constraint on the designer, who does not know the future arrival process when deciding how to direct an arriving agent. The second requirement ensures that the system process encodes the physical breakage process described above. The final requirement ensures that agents can use and fix the machine only when it is working or broken, respectively.

The following result is an analog of Proposition \ref{pr:physfluid} for the sequential arrival environment.  It establishes that every admissible system process induces a reduced-form satisfying the physical constraints.  Further, all triples $(R, P, Q)$ satisfying the physical constraints can be achieved via some Markovian mechanism.  As a result, the physical constraints are necessary and sufficient for feasibility of an allocation.

\begin{proposition}\label{pr:phys}\begin{enumerate}
\item Fix a stationary process that satisfies Assumption~\ref{as:poisson}. The reduced form $(R, P, Q)$ induced by this process satisfy \eqref{eq:balanceConstraint} and \eqref{eq:simplexConstraint} for all $\theta \in \Theta.$

\item Suppose that the triple $(R, P, Q)$ satisfies \eqref{eq:balanceConstraint} and \eqref{eq:simplexConstraint} for all $\theta \in \Theta.$  Then there exists a Markovian mechanism that induces the reduced-form $(R, P, Q)$.
\end{enumerate}
\end{proposition}

\subsection{Connection to incentivized bandit exploration}\label{subsec:poissonApplication}
The sequential interpretation described in Section \ref{subsec:poissonMicrofoundation} can be used to bridge our reduced-form setting and existing models of recommender system design, in particular the setting of KMP.  These models typically employ an incentivized bandit exploration framework, in which a sequence of short-run agents are each recommended one of several bandit arms to pull, with the designer observing the outcome of each pull.  
In this section, we describe a model of bandit exploration which can be directly compared to the setting of KMP and then explain how to reduce it to our collective upkeep model. 

Consider an environment with a restless bandit arm, which each arriving agent can choose to pull or not.  The payoff to pulling the arm is controlled by a state variable $\omega \in \Omega,$ which renews at random times following a Poisson process with rate $\rho$ (whether or not the arm is pulled). When renewal occurs, a new state is drawn from a probability distribution $\pi$ over states, independent of the history of the system.  The designer does not directly observe the state process.  However, whenever an agent pulls the arm, its current state is revealed to the designer.  Additionally, the designer, but not the agents, observes when renewals occur.

An agent of type $\theta$ who pulls the arm in state $\omega \in \Omega$ obtains a type- and state-dependent payoff $v(\omega,\theta)$, where $v\in L^1(\pi\times\lambda)$.  If the agent does not pull the arm, he receives a payoff of 0.  
When an agent arrives, the designer may offer a recommendation to pull or avoid the arm, which the agent is free to follow or ignore.

This setting is a variant of KMP, with three key distinctions. First, the arm's state renews periodically, making it a restless arm. Second, the agents may have heterogeneous preferences. Third, while in KMP calendar time is informative about the designer's knowledge of the arm and is tracked by agents, in our model agents view themselves as arriving in a steady state and calendar time is uninformative. 

We now explain how to reduce this incentivized bandit problem to our collective upkeep model. The machine is a metaphor for the designer's knowledge of the arm's state. It is working when the designer knows the arm's state. Using a working machine means that an agent receives an optimal recommendation tailored to his type and the state of the machine. 

A working machine breaks down when the state of the arm renews. To fix a broken machine, some agent must pull the arm to reveal its state to the designer. Therefore, upkeep of the machine takes the form of periodic exploration of the bandit arm.

Let $\uw,\cst:\Theta\to\bbr$ be given by
\[\uw(\theta) = \int \left(v(\omega, \theta)\right)^+ \, \pi(\td \omega), \quad
\cst(\theta)= -\int v(\omega, \theta) \, \pi(\td \omega).\]
The quantity $u(\theta)$ is the agent's ex-interim payoff from observing the arm's state and making an optimal decision to pull it or not. It is precisely the expected utility she obtains from using a working machine. Meanwhile, the quantity $-c(\theta)$ is the agent's ex-interim payoff from pulling the arm when its state is unknown.  In other words,  the expected cost she incurs fixing a broken machine is $c(\theta)$

In this interpretation, agents cannot distinguish between usage and contribution, since both involve pulling the arm.  As a result, the ex-interim participation constraint of our reduced-form setting is equivalent to an ex-post obedience requirement that each agent must prefer to obey a recommendation by the designer to pull the arm.  

KMP observe that when homogeneous agents do not know their order of arrival in a finite-horizon model, obedience constraints are nonbinding for a utilitarian designer.  In our setting, agents arrive in steady state and calendar time is uninformative about the state of the arm and the designer's knowledge.  A similar conclusion therefore obtains: The first-best outcome is achievable when agents are homogeneous.  However, when agents have heterogeneous preferences, incentive constraints may bind for some agents. Our results characterize how the informativeness of recommendations to different agents should be adjusted to restore obedience.

\section{Proofs}\label{se:proofs}
\subsection{Proof of Proposition~\ref{pr:fb}}
We first observe that the designer's problem is equivalent to the relaxed problem in which the balance condition is enforced as the inequality constraint
\[\rho \cdot Q \leq \int P \, \td\lambda.\]
Indeed, any mechanism which satisfies this inequality strictly can be improved by lowering $P(\theta)$ for a positive measure of types without violating the simplex bounds or the relaxed balance condition.  Hence any optimum of the relaxed problem satisfies the balance condition as an equality.  For the remaining of the proof, we pass to the relaxed problem. 

Let $C^S$ be the set of all $(P, Q)$ that satisfy the simplex bounds.  For each $Q,$ let $\bar{C}^S(Q) = \{P \ : \ (P, Q) \in C^S\}.$ Consider the Lagrangian relaxation of the designer's problem with respect to the balance condition. Let $y$ be the dual variable. The Lagrangian is given by~\eqref{lagrangefb} with the domain $C^S\times \bbr_+$.
The problem \[\text{maximize}_{P \in \bar{C}^S(Q)} \scrl(Q,P;y)\] factorizes over types and has the solution $P(\theta)=(1-Q) \cdot \one_{[c(\theta)<y]}$. Let 
\[\ell^\fb(Q;y)=\max_{P \in \bar{C}^S(Q)} \scrl(Q,P;y)=Q \cdot  (\bar u - \rho \cdot y)+(1-Q)\int (c-y)^+~ \td \lambda\] 
be the reduced Lagrangian with domain $[0,1]\times \bbr_+$.
Since the domain $Q\in [0,1]$ is compact, it follows from Sion's minimax theorem that the primal problem admits an optimal solution and strong duality holds.
The dual Lagrange function is given by
\begin{equation}\label{dualfb}\varphi^{\fb}(y)= \max_{Q\in [0,1]} \ell^\fb(Q;y)=
\max\left(\bar u - \rho \cdot y,\int (y-c(\theta))^+\lambda(\td \theta)\right).\end{equation}
Since the function $y\mapsto\bar u-\rho \cdot y$ is decreasing while the function $y\mapsto \int (y-c)^+\, \td\lambda$ is nondecreasing, the dual problem $\min_{y \geq 0}\varphi^{\fb}(y)$ attains its optimum at $y^\fb$, and its value is $W^\fb$  as stated. Finally, the balance condition  implies~\eqref{qfb}.

\subsection{Proof of Theorem~\ref{th:obedience}}
As in the proof of Proposition \ref{pr:fb}, we pass to the relaxed problem in which the balance condition is enforced as an inequality.  Let $C^{SP}$ be the set of all $(P, Q)$ that satisfy the simplex bounds and participation constraints ~\eqref{eq:participationConstraintModified}.  For each $Q,$ let $\bar{C}^{SP}(Q) = \{P \ : \ (P, Q) \in C^{SP}\}.$ Consider the Lagrangian relaxation of the designer's problem with respect to the balance constraint. Let $y$ be the dual variable. The Lagrangian is given by \eqref{lagrangefb} with domain restricted to $(P, Q)\in C^{SP}$. 

For every $(Q,y)$, the problem \[\text{maximize}_{P \in \bar{C}^{SP}(Q)} \scrl(Q,P;y)\] factorizes over types and has the solution $\pi(\theta;Q,y)=\min(1 - Q, Q \cdot\nu(\theta)) \cdot \one_{[\theta>y]}$. The function $\ell(Q,y)$ in the theorem statement is therefore the reduced Lagrangian:
\[\ell(Q; y)=\max_{P \in \bar{C}^{SP}(Q)} \mathscr{L}(Q, P; y).\]
Since the domain $Q\in [0,1]$ is compact, Sion's minimax theorem implies that strong duality holds, the primal problem
\[\max_{Q \in [0, 1]} \min_{y \geq 0} \ell(Q; y)\]
admits an optimum, and the optimal value is given by $W^\ob$, as stated. If, further, the dual problem admits a (finite) solution $y^\ast$, then strong duality ensures that the mechanism $(P^\ob, Q^\ob)$ with $P^\ob(\theta) = \pi(\theta; Q^\ob ,y^\ob)$ solves the designer's problem.

Suppose instead that the dual problem admits no optimum, i.e., $y^* = \infty$. We claim that $\rho \cdot Q \ge \int_\theta \min(1 - Q, Q \cdot\nu(\theta))\, \lambda(\td\theta)$ for every $Q$. Indeed, if there existed some $Q\in [0,1]$ such that $\rho \cdot Q < \int_\theta \min(1 - Q, Q \cdot \nu(\theta)) \, \lambda(\td\theta)$, then there would exist a feasible mechanism satisfying Slater's condition.  In that case, the dual problem would admit an optimum \parencite[Theorem 3.2.8]{borweinConvexity}, a contradiction.  
This claim implies that any $(P, Q)\in C^{SP}$ satisfying the balance condition must satisfy $P(\theta)=\min(1 - Q, Q \cdot \nu(\theta))$ for every $\theta$. Hence $(P^\ob, Q^\ob)$ with $P^\ob(\theta) = \pi(\theta; Q^\ob ,y^\ob)= \min(1 - Q, Q \cdot \nu(\theta))$ must solve the designer's problem.

It remains to prove that $y^\ob\ge y^\fb$.  If $y^\ob = \infty$ then this inequality trivially holds as a strict inequality, so suppose that $y^\ob < \infty.$ Let $\varphi^\fb$ and $\varphi^\ob$ be the dual Lagrange functions of the two problems, with $\varphi^\fb$ as defined in \eqref{dualfb} and \[\varphi^{\ob}(y) = \max_{Q \in [0, 1]} \ell(Q; y).\] Then $\varphi^\fb(y)= \varphi^\ob(y)$ for every $y\le y^\fb$,  because for $y\le y^\fb$ the maximum in~\eqref{dualfb} is achieved by $(P, Q) = (0, 1)$, which satisfies the participation constraints. Therefore, for every $y<y^\fb$ we have that
\[\varphi^\ob(y) = \varphi^\fb(y) > \varphi^\fb(y^\fb) = \varphi^\ob(y^\fb),\]
where the strict inequality follows from the fact that $\varphi^\fb$ is uniquely maximized at $y^\fb$.
Since $y^\ob$ is the solution to the dual problem $\min_{y \geq 0} \varphi^\ob(y)$, it follows that $y^\ob\ge y^\fb$ as desired.  Finally, if $W^\ob < W^\fb$ and $y^\ob < \infty,$ then $\varphi^\ob(y^\ob)<\varphi^\fb(y^\fb)$. Since $\varphi^\ob(y^\fb)=\varphi^\fb(y^\fb)$, this implies that $y^\fb\neq y^\ob$, as desired.

\subsection{Proof of Proposition \ref{prop:orderedTypesParticipation}}
Let $(Q^\ob, P^\ob)$ be the mechanism characterized in Theorem \ref{th:obedience}, with $y^\ob$ the corresponding solution to the dual problem.  Let $\Theta^\ob$ be the set of types with active participation constraints under this mechanism. 
By definition, $\theta\in\Theta^\ast$ if and only if $P^\ob(\theta) = Q^\ob \cdot \nu(\theta).$ If $Q^\ob = 0,$ then $\Theta^\ast=\bfr_-$ is an interval. Going forward, we assume that $Q^\ob > 0.$

Theorem \ref{th:obedience} implies that $\theta \in \Theta^\ob$ if and only if $c(\theta) < y^\ast$ and $\nu (\theta) \leq (1 - Q^\ob)/Q^\ob$, i.e., $\Theta^\ast$ is the intersection of $I_r=c^{-1}([0, y^\ast))$  and $I_l=\nu^{-1}([0,(1 - Q^\ob)/Q^\ob])$.  Since $c$ and $\nu$ are monotone functions, $I_r$ and $I_l$ are both intervals. Therefore, $\Theta^\ast$ is itself an interval.

Let $\Theta^\fb$ be the set of types that contribute under the mechanism characterized in Proposition \ref{pr:fb}. We claim that $\Theta^\fb\cap \Theta^\ast\neq\emptyset$. Indeed, in any convex optimization problem, all inactive constraints may be removed without affecting the optimality of the solution.  But first-best welfare is achievable in the problem without the participation constraints in $\Theta^\fb$, and so $\Theta^\fb\cap \Theta^\ast = \emptyset$ would contradict our hypothesis. 

Recall that $c(\theta)=-\theta$ under ordered types , so that $\Theta^\fb=(\theta^\fb,0)$ with $\theta^\fb=-y^\fb$. We claim that $\theta^\fb\in \Theta^\ast$. Indeed,  from Theorem \ref{th:obedience} we know that $y^\ob > y^\fb$ and therefore $y^\fb \in [0, y^\ob)$ and $\theta^\fb=-y^\fb \in I_r$. Meanwhile, $\Theta^\fb\cap \Theta^\ast\neq\emptyset$ implies that $\Theta^\fb\cap I_l\neq\emptyset$. Since $\Theta^\fb$ is an interval with lower endpoint $\theta^\fb$ and $I_l$ is a left half-line, it follows that $\theta^\fb\in I_l$.

\subsection{Proof of Proposition \ref{prop:compStatYP}}
Let $(Q^\ob, y^\ob)$ be as defined in Theorem \ref{th:obedience}.  We first prove that if $\mathbb{E}_\lambda \nu = \infty,$ then $Q^\ob \in (0, 1)$ and $y^\ob < \infty,$ so that in particular $Q^\ob \in \argmax_{Q \in (0, 1)} \ell(Q; y^\ob).$  The balance condition implies the inequality $Q^\ob/(1 - Q^\ob) \leq 1/\rho$, so that necessarily $Q^\ob < 1.$  Next, the monotone convergence theorem implies that
\[\lim_{Q \downarrow 0} \int \min(\nu(\theta), (1 - Q)/Q) \, \lambda(\td \theta) = \mathbb{E}_\lambda \nu = \infty.\]
Hence for $Q > 0$ sufficiently small, the mechanism $(Q, P)$ with $P(\theta) = \min(1 - Q, Q \cdot \nu(\theta))$ lies in $C^\ob$ and satisfies the balance condition as a strict inequality, verifying Slater's condition and therefore (as established in the proof of Theorem \ref{th:obedience}) implying that $y^\ob < \infty$.  

By construction, this mechanism satisfies participation for all agents and therefore delivers non-negative aggregate welfare.  Lowering $P(\theta)$ for a positive measure of types until the balance condition holds with equality therefore yields a feasible mechanism generating positive payoffs for a positive measure of types, i.e., generating positive aggregate welfare.  Since any feasible mechanism satisfying $Q = 0$ yields zero welfare, it must be that $Q^\ob > 0.$

We now establish the claimed limits.  Since the reduced Lagrangian $\ell$ is a concave function of $Q,$ its one-sided derivatives exist everywhere.  Its right derivative may be calculated by writing $\ell$ as
\begin{align*}
\ell(Q; y) = \ & Q \cdot (\bar{u} - \rho \cdot y) + Q \cdot \int \nu(\theta)  \cdot (y - c(\theta))^+ \, \lambda(\td\theta) \\
& + \int \one_{[\nu \geq (1 - Q)/Q]}(\theta) \cdot \Big(1 - (\nu(\theta) + 1) \cdot Q\Big) \cdot (y - c(\theta))^+ \, \lambda(\td\theta),
\end{align*} 
Since $(1 - Q)/Q$ is decreasing and continuous,
 \[\frac{d}{d_+Q}\one_{[\nu \geq (1 - Q)/Q]}(\theta) = 0\]
 everywhere, meaning that
 \begin{align*}
\frac{\partial \ell}{\partial_+ Q} = \ & \bar{u} - \rho \cdot y + \int \nu(\theta)  \cdot (y - c(\theta))^+ \, \lambda(\td\theta) \\
\ & - \int \one_{[\nu \geq (1 - Q)/Q]}(\theta) \cdot (\nu(\theta) + 1) \cdot (y - c(\theta))^+ \, \lambda(\td\theta).
\end{align*}
Regrouping terms allows this derivative to be equivalently written
\begin{align*}
\frac{\partial \ell}{\partial_+ Q} = \ & \bar{u} - \int \one_{[c \leq y]}(\theta) \cdot\Big(\one_{[\nu < (1 - Q)/Q]}(\theta) \cdot u(\theta) - \one_{[\nu \geq (1 - Q)/Q]}(\theta) \cdot c(\theta)\Big) \, \lambda(\td\theta) \\
\ & - y \cdot \left(\rho + \text{Pr}_\lambda(c \leq y \land \nu \geq (1 - Q)/Q) - \int  \one_{[c \leq y \land \nu < (1 - Q)/Q]}(\theta) \cdot \nu(\theta) \, \lambda(\td\theta) \right).
\end{align*}
The optimality condition $Q^\ob \in \argmax_{Q \in (0, 1)} \ell(Q; y^\ob)$ implies the one-sided first-order condition $\frac{\partial \ell}{\partial_+ Q}(Q^\ob; y^\ob) \leq 0.$  Meanwhile, for every $Q > 0$ the balance condition may be rewritten
\[\rho = \int  \one_{[c \leq y \land \nu < (1 - Q)/Q]}(\theta) \cdot  \nu(\theta) \, \lambda(\td\theta) + \frac{1 - Q}{Q} \cdot \Pr(c \leq y \land \nu \geq (1 - Q)/Q).\]
Since the balance condition is satisfied at $(Q^\ob, y^\ob)$ and $Q^\ob > 0,$ the first-order condition implies
\begin{align*}
&  \bar{u} - \int \one_{[c \leq y^\ob]}(\theta) \cdot\Big(\one_{[\nu < \nu^\ob]}(\theta) \cdot u(\theta) - \one_{[\nu \geq \nu^\ob]}(\theta) \cdot c(\theta)\Big) \, \lambda(\td\theta)\\
\leq \ & y^\ob \cdot (\nu^\ob + 1) \cdot \text{Pr}_\lambda(c \leq y^\ob \land \nu \geq \nu^\ob),
\end{align*}
where $\nu^\ob = (1 - Q^\ob)/Q^\ob.$  Since $c > 0,$ this bound in turn implies the weaker bound
\[\int \one_{[c > y^\ob]}(\theta) \cdot u(\theta) \, \lambda(\td\theta) \leq y^\ob \cdot (\nu^\ob + 1) \cdot \text{Pr}_\lambda(\nu \geq \nu^\ob).\]

Now, suppose that $\lim_{v \rightarrow \infty} v \Pr_\lambda(\nu > v) = 0.$  Since $\lim_{\rho \rightarrow \infty} Q^\ob = 0,$ it must be that $\lim_{\rho \rightarrow \infty} \nu^\ob = \infty$ and therefore $\lim_{\rho \rightarrow \infty}  (\nu^\ob + 1) \cdot \text{Pr}_\lambda(\nu \geq \nu^\ob) = 0.$  Let $\underline{y} = \liminf_{\rho \rightarrow \infty} y^\ob ,$ and suppose by way of contradiction that $\underline{y} < \infty.$  Then
\[\liminf_{\rho \rightarrow \infty} y^\ob \cdot (\nu^\ob + 1) \cdot \text{Pr}_\lambda(\nu \geq \nu^\ob) = 0.\]
Meanwhile, viewed as a function of $y$, $\one_{[c > y]}(\theta)$ is an indicator function on an open set, meaning that it is lower semicontinuous.  Hence 
\[\liminf_{\rho \rightarrow \infty} \one_{[c > y^\ob]} \geq \one_{[c > \underline{y}]}.\]
Fatou's lemma therefore implies
\[\liminf_{\rho \rightarrow \infty} \int \one_{[c > y^\ob]}(\theta) \cdot u(\theta) \, \lambda(\td\theta) \geq \int \one_{[c > \underline{y}]}(\theta) \cdot u(\theta) \, \lambda(\td\theta) \]
Since $c$ has full support on $\bfr_+,$ the latter integral must be positive when $\underline{y} < \infty,$ a contradiction.

Meanwhile, the left derivative of $\ell$ may be calculated by writing the reduced Lagrangian as
\begin{align*}
\ell(Q; y) = \ & Q \cdot (\bar{u} - \rho \cdot y) + \int \one_{[\nu \leq (1 - Q)/Q]}(\theta) \cdot \Big((\nu(\theta) + 1) \cdot Q - 1\Big)  \cdot (y - c(\theta))^+ \, \lambda(\td\theta) \\
& + (1 - Q) \int (y - c(\theta))^+ \, \lambda(\td\theta).
\end{align*}
Since $(1 - Q)/Q$ is continuous and decreasing, 
\[\frac{d}{d_-Q}\one_{[\nu \leq (1 - Q)/Q]}(\theta) = 0\]
everywhere, meaning that
\begin{align*}
\frac{\partial \ell}{\partial_- Q} = \ & \bar{u} - \rho \cdot y + \int \one_{[\nu \leq (1 - Q)/Q]}(\theta) \cdot (\nu(\theta) + 1)  \cdot (y - c(\theta))^+ \, \lambda(\td\theta) \\
\ & - \int (y - c(\theta))^+ \, \lambda(\td\theta).
\end{align*}
This derivative may be equivalently written
\begin{align*}
\frac{\partial \ell}{\partial_- Q} = \ & \bar{u} - \int \one_{[c < y]}(\theta) \cdot \Big(\one_{[\nu \leq (1 - Q)/Q]}(\theta) \cdot u(\theta) - \one_{[\nu > (1 - Q)/Q]}(\theta) \cdot c(\theta)\Big) \, \lambda(\td\theta) \\
\ & - y \cdot \left(\rho + \text{Pr}_\lambda(c < y \land \nu > (1 - Q)/Q) - \int  \one_{[c < y \land \nu \leq (1 - Q)/Q]} \cdot \nu(\theta) \, \lambda(\td\theta) \right).
\end{align*}
The optimality condition $Q^\ob \in \argmax_{Q \in (0, 1)} \ell(Q; y^\ob)$ implies the one-sided first-order condition $\frac{\partial \ell}{\partial_- Q}(Q^\ob; y^\ob) \geq 0.$  Meanwhile, for every $Q > 0$ the balance condition may be rewritten
\[\rho = \int  \one_{[c < y \land \nu \leq (1 - Q)/Q]} \cdot \nu(\theta) \, \lambda(\td\theta) + \frac{1 - Q}{Q} \cdot \Pr(c < y \land \nu > (1 - Q)/Q).\]
Since the balance condition is satisfied at $(Q^\ob, y^\ob)$ and $Q^\ob > 0,$ the first-order condition implies
\begin{align*}
&  \bar{u} - \int \one_{[c < y^\ob]}(\theta) \cdot \Big(\one_{[\nu \leq \nu^\ob]}(\theta) \cdot u(\theta) - \one_{[\nu > \nu^\ob]}(\theta) \cdot c(\theta)\Big) \, \lambda(\td\theta)\\
\geq \ & y^\ob \cdot (\nu^\ob + 1) \cdot \text{Pr}_\lambda(c < y^\ob \land \nu > \nu^\ob).
\end{align*}
Let $\bar{c} = \int c \, \td\lambda.$  Then the left-hand side of this inequality is in turn bounded above by $\bar{u} + \bar{c},$  so that
\[\bar{u} + \bar{c} \geq y^\ob \cdot (\nu^\ob + 1) \cdot \text{Pr}_\lambda(c < y^\ob \land \nu > \nu^\ob).\]
Note that $\bar{u} + \bar{c}$ is independent of $Q^\ob$, $y^\ob$, and $\rho.$

Now, suppose that $\nu$ and $c$ are negatively associated under $\lambda$ and $\lim_{v \rightarrow \infty} v \Pr_\lambda(\nu > v) = \infty$.  The first condition implies the inequality
\[\text{Pr}_\lambda(c < y^\ob \land \nu > \nu^\ob) \geq \text{Pr}_\lambda(c < y^\ob) \cdot \text{Pr}_\lambda(\nu > \nu^\ob),\]
yielding the bound
\[\bar{u} + \bar{c} \geq y^\ob \cdot \text{Pr}_\lambda(c < y^\ob) \cdot (\nu^\ob + 1) \cdot \text{Pr}_\lambda(\nu > \nu^\ob).\]
Since $\lim_{\rho \rightarrow \infty} Q^\ob = 0,$ it must be that $\lim_{\rho \rightarrow \infty} \nu^\ob = \infty.$  Hence $\lim_{v \rightarrow \infty} v \Pr_\lambda(\nu > v) = \infty$ implies that $\lim_{\rho \rightarrow \infty}  (\nu^\ob + 1) \cdot \text{Pr}_\lambda(\nu > \nu^\ob) = \infty.$  Suppose by way of contradiction that $\limsup_{\rho \rightarrow \infty} y^\ob > 0.$  Since $c$ has full support on $\bfr_+,$ the expression $y \cdot \Pr_\lambda(c < y)$ is an increasing function of $y$ which is positive for all $y > 0$.  Hence $\limsup_{\rho \rightarrow \infty} y^\ob \cdot \Pr_\lambda(c < y^\ob) > 0.$
It follows that
\[\limsup_{\rho \rightarrow \infty} y^\ob \cdot \text{Pr}_\lambda(c < y^\ob) \cdot (\nu^\ob + 1) \cdot \text{Pr}_\lambda(\nu > \nu^\ob) = \infty.\]
But this limit must also be bounded above by $\bar{u} + \underline{u} ,$ which is finite given that $u, c \in \mathscr{L}^1(\lambda),$ a contradiction.

\subsection{Proof of Proposition \ref{prop:compStatYPOrdered}}
Let $f = \lim_{\varepsilon \rightarrow 0} Pr_\lambda(c < \varepsilon)/\varepsilon$, and assume that this limit exists.  Also let $u_+(0) = \limsup_{t \rightarrow 0} u(t)$ and  $u_-(0) = \liminf_{t \rightarrow 0} u(t)$.
\begin{lemma}
If $u_-(0) > 0$ and $u_+(0) < \infty,$ then $\limsup_{v \rightarrow \infty} v \Pr_\lambda(\nu > v) \leq u_+(0) \cdot f$ and $\liminf_{v \rightarrow \infty} v \Pr_\lambda(\nu > v) \geq u_-(0) \cdot f$.
\end{lemma}
\begin{proof}
For any $v > 0,$ let $\chi(v) = -\inf\{\theta \ : \ \nu(\theta) > v\}$, with the convention that $\chi(v) = 0$ if $v \geq \sup_\Theta \nu(\theta).$  Note that $\chi$ is nonincreasing given that $\nu$ is increasing.  The assumption $u_-(0) > 0$ implies that $\nu(\theta) = -u(\theta)/\theta$ is unbounded above, hence  $\chi(v) \in (0, \infty)$ for sufficiently large $v > 0$ and $\lim_{v \rightarrow \infty} \chi(v) = 0$.  Going forward, we restrict attention to $v$ large enough that $\chi(v) < \infty.$

Fix $\epsilon \in (0, 1)$ small enough that $\chi((1 - \epsilon) \cdot  v) < \infty.$  For all $v,$ the bound $c < (1 - \epsilon) \cdot \chi(v)$ implies that $\nu > v,$ while $\nu > v$ implies that $c < (1 + \epsilon) \cdot \chi(v)$.  Hence
\[{\textstyle \Pr_\lambda(c < (1 - \epsilon) \cdot \chi(v)) \leq \Pr_\lambda(\nu > v) \leq \Pr_\lambda(c < (1 + \epsilon) \cdot \chi(v)).}\]
Equivalently,
\[{\textstyle (1 - \epsilon) \cdot v \cdot \chi(v)\cdot \frac{\Pr_\lambda(c < (1 - \epsilon) \cdot \chi(v))}{(1 - \epsilon) \cdot \chi(v)} \leq v \cdot \Pr_\lambda(\nu > v) \leq (1 + \epsilon) \cdot v \cdot \chi(v) \cdot \frac{\Pr_\lambda(c < (1 + \epsilon) \cdot \chi(v))}{(1 + \epsilon) \cdot \chi(v)}.}\]
Now let $v \rightarrow \infty$. Since $\chi$ is nonincreasing in $v$ and satisfies $\lim_{v \rightarrow \infty} \chi(v) = 0,$ we have
\[\lim_{v \rightarrow \infty} \frac{\Pr_\lambda(c < (1 - \epsilon) \cdot \chi(v))}{(1 - \epsilon) \cdot \chi(v)} = \lim_{v \rightarrow \infty} \frac{\Pr_\lambda(c < (1 + \epsilon) \cdot \chi(v))}{(1 + \epsilon) \cdot \chi(v)} = f.\]
Meanwhile, $\nu(-(1 + \epsilon) \cdot \chi(v)) < v < \nu(-\chi(v)/(1 + \epsilon))$ for all $v,$ and so
\[\frac{v \cdot \chi(v)}{1 + \epsilon} > \frac{\nu(-(1 + \epsilon) \cdot \chi(v)) \cdot \chi(v)}{1 + \epsilon} = \frac{u(-(1 + \epsilon) \cdot \chi(v))}{(1 + \epsilon)^2},\]
implying
\[\liminf_{v \rightarrow \infty} \frac{v \cdot \chi(v)}{1 + \epsilon} \geq \frac{u_-(0)}{(1 + \epsilon)^2}.\]
Since $u_-(0)$ is positive and finite by assumption, it follows that
\[{\textstyle \liminf_{v \rightarrow \infty} v \cdot \Pr_\lambda(\nu > v) \geq \frac{f \cdot u_-(0)}{(1 + \epsilon)^2}.}\]
A very similar calculation yields the bound
\[{\textstyle \limsup_{v \rightarrow \infty} v \cdot \Pr_\lambda(\nu > v) \leq  \frac{f \cdot u_-(0)}{(1 - \epsilon)^2}.}\]
Since these bounds hold for any $\epsilon \in (0, 1)$, taking $\epsilon \rightarrow 0$ yields the claimed bounds.
\end{proof}
Combining this result with Proposition \ref{prop:compStatYP} immediately implies the desired result in the first case.  As for the second case, recall that under ordered types, $c$ is decreasing in $\theta$ while $\nu$ is increasing, so that $\nu$ and $c$ are mechanically negatively associated under ordered types.  Proposition \ref{prop:compStatYP} therefore also implies the result for the second case. 

\subsection{Proof of Proposition~\ref{qdaggerroh}}


We prove the Proposition using the following lemma, which may be useful in other situations to establish comparative statics of saddle points. 
\begin{lemma}\label{le:saddlesubmodular}Let $X$ be a sublattice of $\bfr_{+}\times\bfr\times \bfr$ and $h:X\to \bfr$ be submodular and nonincreasing in its final argument. Then $\argmax_{x} \min_{y} x\cdot h(x,y,z)$ is nonincreasing\footnote{Monotonicity is in the induced set ordering \parencite[Section 2.4]{topkis1998supermodularity}.}.\end{lemma}
\begin{proof}Let $H(y, z) = \min_y h(x, y, z),$ and consider the objective function
\[{\textstyle g(x, z)=\min_{y} x\cdot h(x,y,z)=x\cdot H(x, z).}\]
Since $h$ is submodular, it follows that $H$ is as well \parencite[Theorem 2.7.6]{topkis1998supermodularity}. Since $H$ is also nonincreasing in $z$, Claim \ref{cl:super} below implies that $g$ is also submodular. Its maximizers are therefore nonincreasing in $z$.\end{proof}
\begin{claim}\label{cl:super}
Let $X$ be a sublattice of $\bfr_+\times\bfr$ and $f:X\to\bfr$ be submodular (supermodular) and nonincreasing (nondecreasing) in its second component. Then $g:X\to \bfr$ given by $g(x,y)=x \cdot f(x,y)$ is submodular (supermodular).
\end{claim}

\begin{proof}
The claim follows from general results about the product of supermodular functions~\parencite[Lemma 2.6.4]{topkis1998supermodularity}, but we provide a simpler direct proof here as an alternative.

We prove the submodularity case, with supermodularity following from the submodularity result applied to $-f.$  Let $x<x'$. Then 
\[g(x',y)-g(x,y)=(x'-x) \cdot f(x,y)+x \cdot (f(x',y)-f(x,y))\]
is nonincreasing in $y$ since $f$ is nonincreasing in $y$ and submodular. \end{proof} 

To complete the proof of Proposition~\ref{qdaggerroh} using Lemma~\ref{le:saddlesubmodular}, consider the designer's primal problem $\max_{Q \in [0,1]} \min_{y \geq 0} \ell(Q;y)$, with $\ell(Q;y)=Q\cdot h(Q,y,\rho)$ and   
\[h(Q,y,\rho)=\bar u - \rho \cdot y+\int \min\left((1 - Q)/Q,\nu(\theta)\right)(y-\cst(\theta))^+~\lambda(\td\theta).\]
The function $h$ is submodular, because it is a sum of components which are each a product of a nonincreasing function in one variable and a nondecreasing function in another. It is also nonincreasing in $\rho$. Monotonicity of the solution therefore follows from Lemma~\ref{le:saddlesubmodular}.

\subsection{Proof of Theorem~\ref{th:ic}}
As in the proof of Proposition \ref{pr:fb}, we pass to the relaxed problem in which the balance condition is enforced as an inequality.  Let $C^{\ic}$ be the set of all $(R, P, Q)$ that satisfy the simplex bounds along with the participation and IC constraints.  For each $Q,$ let $\bar{C}^{\ic}(Q) = \{(R,P) \ : \ (R,P, Q) \in C^{\ic}\}.$ 

Consider the Lagrangian relaxation of the designer's problem with respect to the balance constraint. Let $y$ be the dual variable. The Lagrangian is given by \eqref{lagrangefb} with domain restricted to $(P, Q)\in C^{SP}$. Let $\ell^\ic(Q,y)$ be the reduced Lagrangian:
\[\ell^\ic(Q; y)=\max_{(R,P) \in \bar{C}^{\ic}(Q)} \mathscr{L}^\ic(R,P,Q; y),\]
where 
\begin{align*}\mathscr{L}^\ic(R, P, Q; y) = \int \left(R(\theta) \cdot \nu(\theta)-P(\theta)\right) \cdot c(\theta)\, \lambda(\td\theta)+y \cdot \int P \, \td\lambda - \rho \cdot Q \cdot y.
\end{align*}

Note that the IC constraints may be rewritten 
 \[ R(\theta) \cdot \nu(\theta) -P(\theta)\ge R(\theta') \cdot \nu(\theta)-P(\theta').\]
These are the familiar IC constraints in a single-good monopoly sale problem with valuations $\nu(\theta)$, probabilities $R(\theta)$ of receiving the good, and payments $P(\theta).$  Maximizing $\mathscr{L}^\ic$ with respect to $(R, P) \in \bar{C}^{IC}(Q)$ is therefore equivalent to solving the monopoly sale problem for a seller who maximizes a weighted sum of buyer welfare and revenue, under the uniform payment bound $P(Q) \leq 1 - Q$. (The remaining simplex bounds $R(\theta)\le Q$ can be normalized to yield a quantity 1 of the good by redefining $\nu$, and so they do not impact the form of the solution.)

As we show in Appendix \ref{singlegood}, the optimal mechanism with bounded payments takes one of two forms.  One possibility is a posted price which delivers the good with the maximum probability, i.e., with $R(Q) = Q$.  This case corresponds to $\bar \nu = \infty$ in the theorem statement.  The second is a menu with two options: pay $1-Q^*$ and receive the item with the maximum probability $Q$, or pay $\hat P$ and receive the item with some probability $\hat R$.  This case corresponds to $\bar \nu < \infty$ in the theorem statement.

We next observe that, since $\ell^\ic$ is the partially-optimized Lagrangian of a convex optimization problem, standard arguments imply that it is a concave-convex function.  (Convexity in $y$ follows from the fact that $\ell^\ic$ is the Lagrangian dual of the designer's problem enforcing a fixed $Q.$  Concavity in $Q$ follows from the fact that $\ell^\ic$ is the value of the convex maximization problem with objective $\mathscr{L}^\ic$, choice variables $(R, P)$, and parameter $Q$.)  Since the domain $Q\in [0,1]$ is compact, Sion's minimax theorem therefore implies that strong duality holds and the primal problem
\[\max_{Q \in [0, 1]} \min_{y \geq 0} \ell(Q; y)\]
admits an optimum $Q^*$. If, further, the dual problem
\[\min_{y \geq 0} \max_{Q \in [0, 1]} \ell(Q; y)\]
admits a (finite) optimum $y^*$, then strong duality ensures that the mechanism $(P^*, R^*)$ that maximizes $\mathscr{L}(R, P, Q^*; y^*)$ solves the designer's problem.

Suppose instead that the dual problem admits no optimum. We claim that $\rho\cdot Q \ge \int P~\td\lambda$ for every mechanism in $C^\ic$. Indeed, if there existed some mechanism $(R, P, Q) \in C^\ic$ satisfying $\rho\cdot Q < \int P \, \td\lambda$, then Slater's condition would be satisfied and the dual problem would admit an optimum, a contradiction.  Therefore, any mechanism in $C^{IC}$ satisfying the balance condition also maximizes seller revenue among all mechanisms in $C^{IC}$. We again appeal to the characterization in Appendix \ref{singlegood} to conclude that in this case an optimal mechanism must have the structure stated in the Theorem.

\subsection{Proof of Proposition~\ref{pr:physfluid}}
Assume first that the system process is stationary and let $(R, P, Q)$ be its reduced form. We first derive the balance condition. If $Q=0$, then the machine is always broken, agents never contribute, and the balance condition is satisfied. So assume that $Q>0$. 

For $T>0$, let $n(T)$ be the number of times the machine breaks down during the interval time interval $[0,T]$. We prove the balance condition by showing that 
\[\rho \cdot Q = \int P \, \td\lambda= \lim_{T\rightarrow\infty} \frac{n(T)}{T} \quad \text{a.s.}\] 

Since $Q>0$ it must be that $\lim_{T\rightarrow\infty}n(T)=\infty$ a.s.\ Indeed, if $n(T)$ were bounded, then from some point onward the machine would never be fixed, else it would eventually break again.  But then from some later point onward the machine would remain broken forever, implying $Q=0$, a contradiction.

Let $L_0\ge 0$ be residual lifespan of the machine at time $0$ (so $L_0=0$ if the machine is broken at time $0$), and for $i\ge 1$ let $L_i$ be the lifespan of the machine after the $i$-th time it is being fixed after time $0$. For $T > 0,$ let $W(T) = \int_0^T \mathbf{1}_{[\omega_t = G]} \, \td t$.
Then for every $T > 0,$
\begin{equation}\label{sandwich} L_1 + \dots + L_{n(T)-1} \le W(T) \le L_0 + L_1 + \dots + L_{n(T)}.\end{equation}
By the strong law of large numbers, $\lim_{n\rightarrow\infty}n^{-1} \sum_{i = 1}^n L_n= 1/\rho$ a.s. Since $\lim_{T\rightarrow\infty}n(T)=\infty$ a.s., it follows from~\eqref{sandwich} that $\lim_{T\rightarrow\infty}W(T)/n(T)= 1/\rho$ a.s. Meanwhile, $\lim_{T\rightarrow\infty}W(T)/T= Q$ a.s.\ by the definition of $Q.$ Therefore $\lim_{T\rightarrow\infty}n(T)/T= Q \cdot \rho$ a.s.

Let $M_0\ge 0$ be the residual contribution quantum at time $0$ (so $M_0=0$ if the machine is working at time $0$), and for $i\ge 1$ let $M_i$ be the contribution quantum after the $i$-th time breakdown subsequent to time $0$.  For $T > 0,$ let $C(T)$ be the total amount of contribution in the time interval $[0,T]$. Then for every $T > 0,$
\begin{equation}\label{sandwichm} M_1+\dots+M_{n(T)-1} \le C(T) \le M_0 +M_1+\dots+M_{n(T)}.\end{equation}
By the strong law of large numbers, $\lim_{n\rightarrow\infty}n^{-1} \sum_{i = 1}^n M_n= 1/\rho$ a.s. Since $\lim_{T\rightarrow\infty}n(T)=\infty$ a.s., it follows from~\eqref{sandwichm} that $\lim_{T\rightarrow\infty}C(T)/n(T) = 1$ a.s. Meanwhile, $\lim_{T\rightarrow\infty}C(T)/T= \int P~\td\lambda$ a.s.\ by the definition of $P.$ Therefore $\lim_{T\rightarrow\infty}n(T)/T= \int P~\td\lambda$ a.s.

We now establish the simplex bounds. At every point in time the realized usage of each agent is bounded by the $1$ if the machine is working and by $0$ if it is broken. Therefore for the population of agents of each type, expected usage is bounded by $Q,$ the probability that the machine is working. In other words, $R(\theta) \leq Q.$  Similarly, the realized contribution of each agent is bounded by the $0$ if the machine is working and by $1$ if it is broken. Therefore for the population of agents of each type, expected contributions are bounded by $1 - Q,$ the probability that the machine is broken. In other words, $P(\theta) \leq 1 - Q.$ 

Assume now that $(R, P, Q)$ satisfies the physical constraints, and consider the Markovian mechanism given by $\sigma(\theta,W)=R(\theta)/Q$ and $\sigma(\theta,B)=P(\theta)/(1 - Q)$, with the convention that $\sigma(\theta, W) = 0$ if $Q = 0$ and $\sigma(\theta, B) = 0$ if $Q = 1$. The machine state process under this mechanism is an alternating renewal process: When the machine becomes working, it remains working for the lifespan with expectation $\mu_W=1/\rho$. When the machine becomes broken, it remains broken for a random time with expectation $\mu_B=1/\int \sigma(\theta,B)~\lambda(\td\theta)$. 

The balance condition implies that $Q/(1 - Q) = \mu_W/\mu_B,$ which may be rearranged to find that the long-run proportion of time that the machine is working is $\frac{\mu_W}{\mu_W+\mu_B}=Q$.  The quantity $Q$ must therefore also be the probability with which the machine is working under the steady state induced by the mechanism. It then follows from definition of $\sigma$ that $R(\theta) = Q \cdot \sigma(\theta, W)$ is the expected usage of an arriving agent of type $\theta$ while $P(\theta) = (1 - Q) \cdot \sigma(\theta, B)$ is the expected contribution of an arriving agent of type $\theta$.

\subsection{Proof of Proposition \ref{pr:phys}}
Assume first that the system process is stationary and let $(R, P, Q)$ be its reduced form. The proof that $(R, P, Q)$ satisfies the balance condition is very similar to the proof of Proposition~\ref{pr:physfluid} and is omitted. To prove the simplex bound, note that the non-anticipation assumption in \eqref{as:poisson} implies the so-called PASTA (``Poison arrivals see time average'') property, under which the probability that an arriving agent of type $\theta$ encounters a working machine equals the steady-state uptime $Q$ \parencite[Section 2.4]{tijms2003first}. As a result, the simplex bounds follow from the fact that the agent cannot use or fix the machine more often than that action is available.

Assume now that $(R, P, Q)$ satisfies the physical constraints, and consider the Markovian mechanism given by $\sigma(\theta,W)=R(\theta)/Q$ and $\sigma(\theta,B)=P(\theta)/(1 - Q)$, with the convention that $\sigma(\theta, W) = 0$ if $Q = 0$ and $\sigma(\theta, B) = 0$ if $Q = 1$. The machine state process under this mechanism is an alternating renewal process with the same expectations as in the proof of Proposition~\ref{pr:phys}.  Therefore, by the same argument as in that proof the reduced form induced by this mechanism is $(R, P, Q).$

\appendix

\section{Monopoly sale problem with bounded transfers}\label{singlegood}
Consider a generalized version of Myerson's single-good monopoly sale problem \parencite{myerson}.  Buyers are heterogeneous with types $\theta \in \Theta$.  They have quasilinear payoffs over allocations and payments, with valuations $\nu : \Theta \to \bfr_+$ for the good.  A mechanism is given by an allocation function $r:\Theta \to [0,1]$ and payment rule $p:\Theta \to\bfr_+$. (We rule out payments to buyers.)  The seller's objective is a weighted average of buyer payoffs and payments of the form \begin{equation}\label{selling-obj}\int \Big( r(\theta) \cdot \nu(\theta) -p(\theta)\Big) \, \mu_1(\td\theta) + \int p \, \td \mu_2\end{equation} for some finite measures $\mu_1,\mu_2$ over $\Theta$. 

A mechanism is incentive compatible (IC) if 
 \[ r(\theta) \cdot \nu(\theta) -p(\theta)\ge r(\theta') \cdot \nu(\theta) -p(\theta')\]
for every $\theta,\theta'\in \Theta$. 
It is individually rational (IR) if $r(\theta) \cdot \nu(\theta) -p(\theta) \geq 0$ for all $\theta \in \Theta$.  The seller chooses a mechanism to maximize \eqref{selling-obj} subject to IC and IR.

 We first observe that, without loss, we may restrict attention to mechanisms satisfying $(r(\theta), p(\theta)) = (r(\theta'), p(\theta'))$ whenever $\nu(\theta) = \nu(\theta').$  Indeed, under any IC mechanism, all types with the same valuation receive the same utility.  If such a mechanism does not offer all such types the same bundle, modify it to pool all types on the bundle involving the highest payment.  This modification does not affect IC or IR and raises the seller's payoff.  Given the restriction to such mechanisms, we may without loss relabel types so that the type space is the set of valuations $\bfr_+.$
 
As usual, the set of IC and IR mechanisms can be identified with the set of nondecreasing allocation functions. (Because the lowest type has valuation 0 for the good, the utility of the lowest type is not a free parameter in this environment.) We may further restrict attention to right-continuous allocation functions.  Indeed, if $r(\nu) < r(\nu+)$ for some $\nu,$ then also $p(\nu) < p(\nu+).$  Additionally, IC implies that the bundle $(r(\nu+), p(\nu+))$ must deliver that type the same utility as $(r(\nu), p(\nu)).$  Replacing that type's bundle with $(r(\nu+), p(\nu+))$ therefore preserves IC and IR, delivers that type the same utility, and raises the seller's revenue from that type.  Going forward, we assume that $r$ is both nondecreasing and right-continuous.
The payment function and buyer payoffs corresponding to each such allocation function $r$ are
\begin{equation}\label{r2pu}\begin{split}& p_r(\nu)= \int_0^\nu v \, \td r(v),\\& U_r(\nu)=\int_0^\nu r(v) \, dv.\end{split}\end{equation}

Standard theory says that the optimal mechanism is a posted price mechanism of the form $r=\one_{[\nu_1,\infty)}$ for some $\nu_1\in \bfr_+$ (the price). The following lemma states an analog of this result for the case that payments are bounded. The proof follows a similar argument as in the corresponding proof for the case of unbounded payments, as in \textcite[Section 2.5]{borgers2015introduction}, with the twist that the set of extreme points of IC mechanisms is more complicated. We also develop the argument more generally, to allow for atoms in the type distributions. 
The argument is not new, but we do not know of a reference. 
\begin{proposition}Consider the problem of maximizing objective~\eqref{selling-obj} over all IC and IR mechanisms such that $0\le p(\nu)\le 1$ for every $\nu\in\bfr_+$. The optimal allocation function is either of the form $r=\one_{[\nu_1,\infty)}$ for some $\nu_1\le 1$ or of the form $r=r_0\one_{[\nu_0,\nu_1)} + \one_{[\nu_1,\infty)}$ for some $0\le \nu_0 \le \nu_1$ and some $0 \le r_0 \le 1$ such that $r_0\nu_0 + (1-r_0) \nu_1= 1$ \end{proposition}
The first mechanism in the lemma is the familiar posted price with price $\nu_1$. The second mechanism can be implemented as a menu with two options: pay $r_0 \cdot \nu_0$ and receive the good with probability $r_0$, or pay $1$ and receive the good for sure.
\begin{proof}The  allocation functions  $r:\bfr_+\to [0,1]$ of IR and IC mechanisms with payments bounded by $1$ are nondecreasing and right-continuous, hence also u.s.c, and satisfy $\int_0^\infty v~\td r(v) \le 1$. We can identify each such allocation with
  a finite measure over $m$ over $\bfr_+$ such that 
  \begin{align}
&m(\bfr_+) \leq 1, \text{ and}\label{mass}\\
&\int \nu~m(\mathrm{d}\nu) \leq 1, \label{first-moment}
\end{align}
so $r$ is the cdf of $m$. 
The set $\calM$ of all these measures is convex, and, by Helly's selection theorem, compact in the weak$^\ast$ topology.

By~\eqref{r2pu}, for every buyer valuation $\nu$, the utility $U_r(\nu)$ is linear and continuous in $r$ and the payment $p_r(\nu)$ is linear and u.s.c.\  in $r$. Therefore, the objective function \eqref{selling-obj} is linear and u.s.c., which implies that the maximum is achieved at an extreme point. The set of extreme points of measures under a finite number of moment conditions is well-understood. (See, for example~\textcite{karr1983extreme}. The argument traces back to~\textcite{karlin1953geometry}.)

The extreme points of $\calM$ for which \eqref{mass} is not saturated correspond to measures with a single atom and allocation function of the form $r=r_1\one_{[\nu_1,\infty)}$ with $r_1<1$ and $r_1\nu_1 \le 1$, implementable by a mechanism that offers a price$r_1\nu_1$ in return for receiving the item with probability $r_1$. This mechanism cannot be optimal under our objective function, because both buyer utility and payments can be increased by offering the item with probability $1$ for the same price. 

The extreme points of $\calM$ for which \eqref{mass} is saturated are probability measures. If \eqref{first-moment} is not saturated, then $m$ is a probability measure with a single atom, corresponding to an allocation of the form $r=\one_{[\nu_1,\infty)}$ with and $\nu_1 \le 1$, as in the proposition. If \eqref{first-moment} is saturated, then $m$ is a probability measure with at most two atoms, corresponding to an allocation of the form 
$r=r_0\one_{[\nu_0,\nu_1)} + \one_{[\nu_1,\infty)}$ for some $0\le \nu_0 \le \nu_1$ and some $0 \le r_0 \le 1$ such that $r_0\nu_0 + (1-r_0) \nu_1= 1$, as in the proposition.
\end{proof}
Example 2 in \textcite{che2000optimal} presents an environment in which the bound on payments is binding and the optimal mechanism is a menu with two options. That example involves an environment in which the seller is uncertain of the buyer's budget as well as their valuation.  However, the solution involves a menu with payments no higher than the lowest budget, and so it must also be the optimal mechanism in the problem where all buyers have the same (low) budget.

\cleardoublepage
\phantomsection

\addcontentsline{toc}{section}{References}
\printbibliography
\newpage


\end{document}